\documentclass[12pt,a4paper]{article}

\usepackage[margin=1in]{geometry} 
\usepackage[utf8]{inputenc} 
\usepackage{authblk} 
\usepackage{hyperref} 
\usepackage[numbers,sort]{natbib} 
\usepackage{color}
\usepackage{array,multirow}

\usepackage{enumitem}
\setlist{itemsep=0pt}

\usepackage{amssymb,amsmath,amsthm}
\usepackage{mathtools}
\usepackage{graphicx}
\usepackage{subcaption}
\usepackage[ruled]{algorithm2e}
\usepackage[capitalise]{cleveref}

\theoremstyle{plain}

\newtheorem{theorem}{Theorem}[section]
\newtheorem{lemma}[theorem]{Lemma}
\newtheorem{proposition}[theorem]{Proposition}

\theoremstyle{remark}
\newtheorem{definition}[theorem]{Definition}

\newtheorem*{remark}{Remark}

\DeclareMathOperator*{\argmin}{argmin}
\DeclareMathOperator*{\argmax}{argmax}

\newcommand{\lb}{\left(}
\newcommand{\rb}{\right)}
\newcommand{\trans}{\mathsf{T}}%

\newcommand{\mcN}{\mathcal{N}}
\newcommand{\mcJ}{\mathcal{J}}

\newcommand{\opt}{\mathrm{opt}}
\newcommand{\Ent}{\mathrm{Ent}}

\newcommand{\Hess}{\mathrm{Hess}}

\newcommand{\sub}{\mathrm{sub}}

\renewcommand{\d}{\mathrm{d}}

\newcommand\solidrule[1][1cm]{\rule[0.5ex]{#1}{.4pt}}
\newcommand\dashedrule{\mbox{%
  \solidrule[2mm]\hspace{2mm}\solidrule[2mm]\hspace{2mm}\solidrule[2mm]}}

\newcommand{\R}{\mathbb{R}}
\newcommand{\E}{\mathbb{E}}

\DeclareMathOperator{\trace}{trace}

\usepackage[normalem]{ulem} 

\usepackage[nice]{nicefrac}

\providecommand{\keywords}[1]
{
  \small
  \textbf{\textit{Keywords---}} #1
}
\begin{document}

\title{Principal feature detection via $\phi$-Sobolev inequalities}

\author{Matthew T.C. Li\footnote{mtcli@mit.edu}, Youssef Marzouk\footnote{ymarz@mit.edu}, Olivier Zahm\footnote{olivier.zahm@inria.fr}}

\date{
Massachusetts Institute of Technology, Cambridge, MA 02139, USA \\
Univ. Grenoble Alpes, Inria, CNRS, Grenoble INP, LJK, 38000 Grenoble, France}
%

\maketitle

\begin{abstract}
 We investigate the approximation of high-dimensional target measures as low-dimensional updates of a dominating reference measure. This approximation class replaces the associated density
with the composition of: (i) a feature map that identifies the leading principal components or features of the target measure, relative to the reference, and (ii) a low-dimensional profile function.
When the reference measure satisfies a subspace $\phi$-Sobolev inequality, we construct a computationally tractable approximation that yields certifiable error guarantees with respect to the Amari $\alpha$-divergences.
Our construction proceeds in two stages. First, for any feature map and any $\alpha$-divergence, we obtain an analytical expression for the optimal profile function. Second, for linear feature maps, the principal features are obtained from eigenvectors of a matrix involving gradients of the log-density. Neither step requires explicit access to normalizing constants. Notably, by leveraging the $\phi$-Sobolev inequalities, we demonstrate that these features universally certify approximation errors across the range of $\alpha$-divergences $\alpha \in (0,1]$.
We then propose an application to Bayesian inverse problems and provide an analogous construction with approximation guarantees that hold in expectation over the data. We conclude with an extension of the proposed dimension reduction strategy to nonlinear feature maps.
\end{abstract}

\keywords{{Amari $\alpha$-divergences},
{Bayesian inference},
{feature detection},
{gradient-based dimension reduction},
{principal components},
{$\phi$-Sobolev inequalities}.
}

\section{Introduction} \label{sec:intro}

Sampling from complex high-dimensional probability measures is a difficult but ubiquitous problem in statistical computation.
Markov chain Monte Carlo (MCMC) algorithms \cite{Roberts_Rosenthal_2004} are widely used for this purpose, but typically scale poorly with dimension \cite{roberts1998optimal,pillai2012,andrieu2022explicit,mangoubi2018dimensionally}.
This issue similarly affects sequential Monte Carlo algorithms \cite{del2006sequential,Snyder_Bengtsson_Bickel_Anderson_2008, Rebeschini_van_Handel_2015}, or algorithms that seek variational approximations of the target measure based on transport maps and normalizing flows \cite{Marzouk_Moselhy_Parno_Spantini_2016,papamakarios2021normalizing}. Presented with this bottleneck, recent research focuses on identifying latent low-dimensional structure in the target probability measure to improve the computational efficiency of these algorithms.

To this end, we focus on dimension reduction for probability measures $\pi$ on $\R^{d}$ given by
$$
 \d\pi(x) \propto \ell(x) \d\mu(x),
$$
where $\mu$ is a reference probability measure on $\R^d$ and $\ell$ is a smooth positive-valued function. Following \cite{ZCLSM22}, we aim to identify low-dimensional structure expressed as a (possibly nonlinear) \emph{feature map} $\varphi_r: \R^d \to \R^r$, with $r \ll d$, such that the probability measure
\begin{equation}\label{eq:pitilde}
 \d\widetilde\pi_r(x) \propto \widetilde\ell_r \left ( \varphi_r(x)  \right ) \d\mu(x) ,
\end{equation}
approximates $\pi$ for some low-dimensional function $\widetilde\ell_r$ defined on $\R^r$. This approximation class thus replaces the high-dimensional function $\ell$ with the composition $\widetilde \ell_r \circ \varphi_r$; $\widetilde \ell_r$ is called the \emph{profile function}. In many cases (as in \cite{ZCLSM22}), the feature map is linear, i.e., $\varphi_r(x) = U_r^\trans x$ for a matrix $U_r \in \R^{d \times r}$ with orthonormal columns, so that $x\mapsto \widetilde \ell_r \left ( U_r^\trans x  \right )$ is a \emph{ridge approximation} of $\ell$.
This approximation format, with linear features,
has been applied successfully, e.g., to accelerate Stein variational gradient descent algorithms \cite{Chen_Ghattas_2020, Liu_Zhu_Ton_Wynne_Duncan_2022}, to improve sampling for Bayesian inference \cite{Cui_Tong_2021,Constantine_Kent_Bui-Thanh_2016,ehre2023certified, cui2022prior},
and to make rare event simulation tractable \cite{uribe2021cross,tong2022large}, to name just a few applications. (See Section~\ref{sec:litreview} for a discussion of these contributions.)

In this paper, we aim to provide a construction of $\widetilde\pi_r$ such that its approximation error, measured with some statistical divergence, is guaranteed to fall below a certain threshold. Naturally, it is also desirable for such a certificate to be non-vacuous, and to approach zero as $r \to d$.
We also wish to elucidate the interplay between the \emph{choice of divergence measure} and the form of suitable---or optimal---feature maps and profile functions. In general, different divergences will penalize different aspects of the approximation~\cite{Liese_Vajda_2006}, and thus have different uses. For instance, the TV, squared Hellinger, and LeCam divergences have use in binary hypothesis testing \cite{PolyanskiyWuTextbook}, while the KL divergence plays a pivotal role in Shannon's information theory.
To our knowledge, a comprehensive analysis of dimension reduction strategies for a range of divergences has not yet been performed.
To address these questions, we consider the broad class of Amari $\alpha$-divergences~\cite{Amari_2009}, parameterized by~$\alpha\in\R$, which are~$\phi$-divergences of the form
$$
 D_\alpha(\pi||\widetilde\pi_r) =
 \int \phi_\alpha\left( \frac{\d\pi}{\d\widetilde\pi_r} \right) \d\widetilde\pi_r ,
 \qquad \phi_\alpha(t)= \frac{t^\alpha -1}{\alpha(\alpha-1)} - \frac{t-1}{\alpha-1} ,
$$
where $\nicefrac{\d\pi}{\d\widetilde\pi_r}$ denotes the Radon--Nikodym derivative of $\pi $ with respect to $\widetilde\pi_r$.
These $\alpha$-divergences enjoy an interesting interpolative property: $\alpha=2$ yields the $\chi^2$-divergence, to a multiplicative factor of half, $\alpha\rightarrow 1$ corresponds to the KL divergence, $\alpha=\nicefrac{1}{2}$ results in four times the squared Hellinger distance, and $\alpha \rightarrow  0$ recovers the \emph{reverse} KL divergence.

When measuring approximation error with the KL divergence ($\alpha=1$), \cite{ZCLSM22} originally derived an analytical expression for the optimal profile function $\widetilde\ell_r=\ell_r^\opt$, given as the conditional expectation of $\ell$. In a second step, $U_r$ is built by minimizing an error bound obtained by assuming that the reference measure $\mu$ satisfies a \emph{(subspace) logarithmic Sobolev inequality}. Notably, the minimizer of that bound is known analytically, and this yields the matrix $U_r$ containing the $r$ dominant eigenvectors of the so-called \emph{diagnostic matrix}
\begin{equation}\label{eq:defH}
 H = \int \nabla\log\ell(x) \nabla\log\ell(x)^\trans \d\pi(x) .
\end{equation}
Later,~\cite{cui4258736scalable,Cui_Tong_2021} proposed a similar methodology for the squared Hellinger distance, meaning $\alpha=\nicefrac{1}{2}$. In particular, the Hellinger-optimal profile function can also be expressed as a conditional expectation, and $U_r$ is also built by minimizing an analogous error bound, but now obtained by assuming the reference measure satisfies a \emph{(subspace) Poincaré inequality}. Intriguingly, the features  $U_r$ which minimize this  upper bound remain as the $r$ leading eigenvectors of the same diagnostic matrix $H$ previously introduced in \eqref{eq:defH}.
We note that this strategy of furnishing the statistical problem with a computationally favorable majorization underpins much of variational inference; for instance, compare to the use of the evidence-based lower bound for optimizing variational auto-encoders~\cite{Kingma2014a,Rezende_Mohamed_Wierstra_2014}.

Our main contribution is to propose a new analysis which encompasses the two previously considered divergences as special instances (i.e., KL and Hellinger). We consider constructing the approximation $\widetilde\pi_r$ to $\pi$ via the solution to
\begin{equation}
\label{eq:loss}
\min_{ \varphi_r: \R^d \to \R^r} \,
\min_{\widetilde\ell_r:\R^r\rightarrow \R_{\geq0}}
D_\alpha(\pi \,||\, \widetilde\pi_r),
\quad \text{with } \widetilde\pi_r \text{ as in \eqref{eq:pitilde}} .
\end{equation}
In Theorem~\ref{thm:pythagorean} we show that for any $\alpha\in\R$ and for any feature map, the optimal profile function $\widetilde\ell_r=\ell_{\alpha,r}^\text{opt}$ admits a closed form expression as a conditional expectation of a function of $\ell$.
Then, for the case of linear feature maps $\varphi_r(x) = U_r^\trans x$, Theorem~\ref{thm:alphaCDR} derives an upper bound for $D_\alpha(\pi||\pi_{\alpha,r}^\text{opt})$, with $\d\pi_{\alpha,r}^\opt \propto \ell_{\alpha,r}^\opt \d\mu$, for the range $\alpha\in [\nicefrac{1}{2},1]$, and Theorem~\ref{thm:alphaCDRext} to the range $\alpha\in (0,\nicefrac{1}{2}]$. By leveraging  improved mathematical tools, our Theorem~\ref{thm:alphaCDRimp} then further sharpens these bounds.
The main implication of all these theorems is the following: for all~$\alpha$ within this interval, the feature matrix $U_r$ minimizing these bounds remains as the leading eigenvectors of the diagnostic matrix in \eqref{eq:defH}. Thus, denoting the $k$-th largest eigenvalue of $H$ by $\lambda_k$, we can certify the optimality of our low-dimensional approximation following
\begin{equation}\label{eq:errorBound1}
 D_\alpha(\pi\,||\,\pi_{\alpha,r}^\text{opt}) \leq \mcJ_\alpha \left( C_\alpha(\mu) \sum_{k=r+1}^d \lambda_k \right) ,
\end{equation}
where $C_\alpha(\mu)\geq0$ is a constant depending only on $\alpha$ and the reference measure $\mu$, and $t\mapsto \mcJ_\alpha(t)$ is a monotone non-decreasing function.
Crucially, the relation in~\eqref{eq:errorBound1} enables selection of the reduced dimension~$r$ in a principled manner, depending on the balance of computational budget and desired accuracy. Moreover, if the spectrum of the diagnostic matrix~$H$ decays sufficiently quickly, our bound demonstrates that the reduced features $U_r^\trans x$ \emph{universally certify} dimension reduction for all divergences corresponding to $0 < \alpha \leq 1$, albeit with bespoke $\ell_r^\opt$ for each specific divergence.

The bounds in our theorems are derived using $\phi$-Sobolev inequalities~\cite{Chafai_2004,Bolley_Gentil_2010}.
This class of functional inequalities inclusively  interpolates the classical logarithmic Sobolev inequality and the Poincaré inequality, encapsulating the tools used in earlier works.
This new perspective allows us to generalize the intriguing conjugacy relation between divergences and their functional inequalities already discovered for the KL divergence and squared Hellinger divergence: dimension reduction for a divergence of order~$\alpha$ is controlled by the corresponding Sobolev inequality for the divergence of order~$\nicefrac{1}{\alpha}$.
We refer to the constant $C_\alpha(\mu)$ appearing in \eqref{eq:errorBound1} as the \emph{subspace $\alpha$-Sobolev constant}.
Note that except for certain reference measures for which $C_\alpha(\mu)$ is known analytically (for instance $C_\alpha(\mu)\leq 1$ for the standard normal measure $\mu$), one only has access to a bound for this constant provided $\mu$ satisfies certain log-concavity assumptions.

We also specialize our application to Bayesian inverse problems for which the target measure is the \emph{posterior distribution} given as $\d\pi^y(x)\propto \ell^y(x)\d\mu(x)$, where $\ell^y(x)$ is (proportional to) the density of the observations $Y$ conditioned on $x$.
In this setting, we seek a feature map $\varphi_r$ which is optimal for the \emph{average} realization of data \emph{viz.},
\[
\min_{ \varphi_r: \R^d \to \R^r} \,
\,\mathbb{E}_Y\,\left[\, \min_{\widetilde\ell_r^Y :\,\R^r \to \R_{\geq 0}} D_\alpha(\pi^Y \,||\, \widetilde\pi_r^Y) \right] ,
\]
where $\d\widetilde\pi_r^y(x) \propto \widetilde\ell_{r}^y(\varphi_r(x))\d\mu(x)$. While Theorem \ref{thm:pythagorean} provides an analytical solution to the inner optimization problem, we introduce Theorem \ref{corr:datafreebound} to derive an upper bound for the resulting error $\mathbb{E}_Y[D_\alpha(\pi^Y || \widetilde\pi_{\alpha,r}^{Y,\opt})]$ for the case of linear feature maps $\varphi_r(x) = U_r^\trans x$. The key result of this theorem is that we can construct $U_r$ with the dominant eigenvectors of the \emph{data-free} diagnostic matrix
\begin{equation}
\label{eq:datafreemtx}
H_\textrm{DF} = \int \E_{Y|x} \left[  \nabla\log\ell^Y(x) \nabla\log\ell^{Y}(x)^\trans \right]  \d\mu(x) ,
\end{equation}
where $\E_{Y|x}[\cdot]$ corresponds to the expectation of the data $Y$ conditioned on $x$,
and still obtain a certificate of optimality.
We highlight that $\E_{Y|x} [  \nabla\log\ell^Y(x) \nabla\log\ell^{Y}(x)^\trans ] $ is the \emph{Fisher information matrix} evaluated at $x$, which often admits a closed form expression (depending on the statistical model for $Y$).
This extends earlier results of~\cite{Cui_Zahm_2021}, who demonstrated such a bound for the averaged KL divergence, and who also recognized the significant computational advantage that $H_\textrm{DF}$ affords as it can be computed before observing the data; hence the terminology data-free.

Lastly, we conclude with a preliminary foray into \emph{nonlinear} dimension reduction. Specifically, we show that our methodology for constructing linear features naturally extends to constructing general nonlinear features of the form $\varphi_r(x) = U_r^\trans \Phi(x)$, for functions $\Phi$ belonging to the class of $\mathcal{C}^1$-diffeomorphisms. In particular, we propose a construction of $\Phi$ such that the corresponding diagnostic matrix provides low-dimensional features, for which $U_r$ can then be constructed with similar certifiable guarantees as before. This extends a similar approach previously considered in \cite{bigoni2022nonlinear} for the purpose of regression.

The remainder of this paper is organized as follows. We define the $\alpha$-divergences and some technical materials in Section~\ref{sec:optllhd}. We then give the analytical expressions for the optimal profile function when conditioning on either linear or nonlinear features. In Section~\ref{sec:alphaCDR}, we specialize the discussion to linear features. After introducing the $\phi$-Sobolev inequalities in Section~\ref{sec:FunctionalInequality}, we derive our certifiable upper bounds for dimension reduction in Section~\ref{sec:FirstBound} and Section~\ref{sec:monotone_ext}. Section~\ref{sec:generalimprove} proposes an improvement on our bounds in the range $\alpha \in (\nicefrac{1}{2}, 1)$ which utilizes a tighter version of the $\phi$-Sobolev inequalities. Finally in Section~\ref{sec:bayesian}, we apply our techniques to Bayesian inverse problems, and in Section~\ref{sec:nonlinear}, we discuss an extension of the proposed methodology for the detection of nonlinear features. We leave a discussion of the connection of the present work to the broader literature to Section~\ref{sec:litreview}.

\section{Optimal profile function for Amari $\alpha$-divergences}
\label{sec:optllhd}

\subsection{Notation} \label{sec:notation}

\textbf{Amari $\alpha$-divergence~\cite{Amari_2009}.}
For any $\alpha\in\R$, let $\phi_\alpha:[0,\infty)\rightarrow\R$ be the convex function defined by
\begin{equation}
\label{eq:amari}
\phi_\alpha(t)
 =
\begin{cases}
 \dfrac{t^\alpha -1}{\alpha(\alpha-1)} - \dfrac{t-1}{\alpha-1} & \alpha\notin\{0, 1\}, \\
 {-\ln t + t-1} & \alpha = 0,  \\
 {t\ln t -t + 1} & \alpha = 1.
\end{cases}
\end{equation}
The expressions for $\alpha=0$ and $\alpha=1$ can also be recovered by taking the respective pointwise limits.
Given two measures\footnote{{Observe that for normalized measures $\int\d\mu=\int\d\nu=1$, the affine term $(t-1)$ in \eqref{eq:amari} can be omitted from the definition of $\phi_\alpha$.}} $\mu$ and $\nu$ such that the Radon--Nikodym derivative~$\nicefrac{\d\nu}{\d\mu}$ exists, the Amari $\alpha$-divergence is defined by $D_\alpha(\nu||\mu)= \int \phi_\alpha( \frac{\d\nu}{\d\mu} ) \d\mu  $. For $\alpha\notin\{0, 1\}$ we have
\begin{equation}\label{eq:DalphaDef}
 D_\alpha(\nu\,||\,\mu) =  \frac{1}{\alpha(\alpha-1)} \left(\int \left( \frac{\d\nu}{\d\mu} \right)^\alpha \d\mu -1\right),
\end{equation}
which is proportional to the $\chi^2$-divergence when $\alpha=2$ and to the squared Hellinger metric when $\alpha = \nicefrac{1}{2}$. The limiting case $\alpha=1$ recovers the Kullback--Leiber (KL) divergence
\[
D_1(\nu\,||\,\mu) =  \int \ln\lb \frac{\d\nu}{\d\mu} \rb \d\nu,
\]
while $\alpha=0$ recovers the \emph{reverse} KL divergence.
In fact, this latter connection is a special instance of the duality relation $D_\alpha(\nu \,||\,\mu) = D_{1-\alpha}(\mu \,||\, \nu)$, which can be verified by direct computation. Amari~\cite{Amari_2009} also shows a distinguished property of this class of divergences: when suitably extended to the manifold of \emph{unnormalized} positive measures, the $\alpha$-divergences are the only $\phi$-divergences which are Bregman divergences. (For \emph{probability} measures, only the KL divergence
is a Bregman divergence.)
Bregman divergences have the fundamental property that optimal predictors are conditional expectations (see \cite{Banerjee_Guo_Wang_2005}) which we shall exploit later.\\

\noindent\textbf{Conditional measures and conditional expectations.}
Given a probability measure $\mu$ on $\R^d$ we denote by $X\sim\mu$ the random vector distributed according to $\mu$, and by
$$
 \E_{X\sim\mu}[f(X)] = \int f(x) \d\mu(x)
$$
the expectation of an integrable function $f:\R^d\rightarrow\R$.
Later, in Section \ref{sec:generalimprove}, we also use the more compact notation $\mu(f)$ to denote the expectation of $f$ under $\mu$. Given a measurable function $\varphi_r:\R^d\rightarrow\R^r$, we denote by $\mu_r$ the pushforward measure of $\mu$ by $\varphi_r$, which is the probability measure of the random vector $\Theta_r=\varphi_r(X)$.
In particular, we have
\begin{equation}\label{eq:MarginalDensity}
 \E_{\Theta_r\sim\mu_r}[f_r(\Theta_r)]
 = \int f_r(\varphi_r(x)) \d\mu(x)
\end{equation}
for any integrable function $f_r:\R^r\rightarrow\R$.
The disintegration theorem (see Chapter 5 in \cite{kallenberg1997foundations}) states that, for $\mu_r$-almost every $\theta_r\in\R^r$, there exists a unique measure $\mu_{\perp|r}(\cdot|\theta_r)$ on the pre-image $\varphi_r^{-1}(\theta_r)=\{x\in\R^d:\varphi_r(x)=\theta_r\}$ such that
\begin{equation}\label{eq:condExp}
 \int f(x)\d\mu(x) = \int \left(\int f(x) \d\mu_{\perp|r}(x|\theta_r) \right) \d\mu_r(\theta_r),
\end{equation}
for any integrable function $f:\R^d\rightarrow\R$.
In other words, $\mu_{\perp|r}(\cdot|\theta_r)$ is the measure of $X$ conditioned on the event $\varphi_r(X)=\theta_r$. The conditional expectation of a function $f:\R^d\rightarrow\R$ given $\varphi_r(X)=\theta_r$ is defined as
$$
 \E_{X\sim\mu}[f(X) \mid \varphi_r(X)=\theta_r] = \int f(x) \d\mu_{\perp|r}(x | \theta_r) ,
$$
which is now a {function of $\theta_r$ that is measurable with respect to the $\sigma$-algebra generated by $\Theta_r$}.
When there is no ambiguity, we write $\mu_{\perp|r}=\mu_{\perp|r}(\cdot|\theta_r)$ and \eqref{eq:condExp} simplifies as
\begin{equation}\label{eq:condExp2}
 \int f\d\mu = \int f \d\mu_{\perp|r}\d\mu_r.
\end{equation}

\begin{remark}[Linear feature map]
A particular case of interest to us is the linear feature map
\begin{equation}\label{eq:linearGr}
 \varphi_r(x) = U_r^\trans x,
\end{equation}
where $U_r\in\R^{d\times r}$ is a matrix with orthogonal columns. Denoting by $U_\perp\in\R^{d\times (d-r)}$ any orthogonal completion to $U_r$ such that $U=[U_r,\,U_\perp]\in\R^{d\times d}$ forms a unitary matrix, the pre-image $\varphi_r^{-1}(\theta_r)$ is the affine subspace $\varphi_r^{-1}(\theta_r) = \{U_r \theta_r + U_\perp \theta_\perp : \theta_\perp\in\R^{d-r}\}$.
In addition, the pushforward measure $\mu_r$ corresponds to the \emph{marginal} distribution of the first $r$ components of $\Theta = U^\trans X$ and, assuming a Lebesgue density for our reference measure $\d\mu(x) \propto \rho(x) \d x$, then $\mu_r$ also admits a Lebesgue density given by
$$
\d \mu_r(\theta_r) = \left(\int \rho( U_r \theta_r + U_\perp \theta_\perp )\d \theta_\perp \right)\d\theta_r = \rho_r(\theta_r)\d\theta_r.
$$
Denoting by $\d x_\perp$ the Lebesgue measure of the subspace $\varphi_r^{-1}(\theta_r)$, the conditional measure
$\mu_{\perp|r}(\cdot|\theta_r)$ in \eqref{eq:condExp}
then has the Lebesgue density
\begin{equation}\label{eq:conditionalDensity}
\d \mu_{\perp|r}(x|\theta_r) = \frac{\rho( U_r \theta_r  + U_\perp U_\perp^\trans x )}{\rho_r(\theta_r)} \d x_\perp ,
\end{equation}
and the conditional expectation of a function $f$ can be written as
$$
 \E_{X\sim\mu}[f(X) \mid \varphi_r(X)=\theta_r] = \int f( U_r \theta_r  + U_\perp \theta_\perp ) \frac{\rho( U_r \theta_r  + U_\perp \theta_\perp )}{\rho_r(\theta_r)} \d \theta_\perp.
$$
\end{remark}

\subsection{Optimal profile function}\label{sec:Optimalprofilefunction}

In this section, we consider approximating the target measure~$\pi$ with measures of the form
\begin{equation}
\label{eq:nonlinpitilde}
\d\widetilde\pi_r \propto \widetilde\ell_r \left ( \varphi_r(x) \right ) \d\mu(x),
\end{equation}
where $\varphi_r : \R^d \to \R^r$ denotes a generic, possibly nonlinear, $r$-dimensional feature map.
Restricting to linear feature maps $\varphi_r(x)=U_r^\trans x$ recovers \eqref{eq:pitilde}.
The next theorem demonstrates that properties of the $\alpha$-divergences let us write an analytical expression for the profile function $\ell_{\alpha,r}^\opt$ that minimizes $\widetilde\ell_r \mapsto D_\alpha(\pi || \widetilde\pi_r)$, for any $\alpha \in \mathbb{R}$ and profile function $\varphi_r$.

\begin{theorem}[Pythagorean-like identity]
\label{thm:pythagorean}
Let $\pi$ and $\mu$ be probability measures such that $\d\pi(x) \propto \ell(x)\d\mu(x)$ for some integrable function $\ell:\mathbb{R}^d\rightarrow\mathbb{R}_{\geq0}$.
Given a measurable function~$\varphi_r : \mathbb{R}^d \rightarrow \R^r$ {and $\alpha\in\R$},  consider the probability measure
\begin{equation}
\label{eq:piopt}
\d\pi_{\alpha,r}^\opt(x) \propto \ell_{\alpha,r}^\opt ( \varphi_r(x) )\d\mu(x),
\end{equation}
where
\begin{equation}
 \label{eq:opt_llhd}
 \ell_{\alpha,r}^\opt(\theta_r) =
 \left\{\begin{array}{ll}
      \mathbb{E}_{X\sim\mu}[ \ell(X)^\alpha \mid \varphi_r(X)=\theta_r ]^{\frac{1}{\alpha}}
      & \alpha\neq0,  \\[5pt]
      \exp( \mathbb{E}_{X\sim\mu}[ \ln\ell(X) \mid \varphi_r(X)=\theta_r ] )
      & \alpha=0.
 \end{array}
 \right.
\end{equation}
Then, for any integrable function $\widetilde\ell_r:\mathbb{R}^r\rightarrow\mathbb{R}_{\geq0}$, the probability measure~$\d\widetilde\pi(x)\propto \widetilde\ell_r( \varphi_r(x))\d\mu(x)$ satisfies the Pythagorean-like identity
\begin{equation}
\label{eq:pythagorean}
D_\alpha( \pi \,||\, \widetilde\pi_r ) =
D_\alpha( \pi \,||\, \pi_{\alpha,r}^\opt\,) + \lb \dfrac{Z_{\alpha,r}}{Z_\pi}\rb^\alpha \, D_\alpha( \pi_{\alpha,r}^\opt \,||\, \widetilde\pi_r) ,
\end{equation}
where $Z_\pi=\int \ell(x)\d\mu(x)$ and $Z_{\alpha,r}= \int \ell_{\alpha,r}^\opt( \varphi_r(x))\d\mu(x)$ are the normalizing constants of $\pi$ and $\pi_{\alpha,r}^\opt$, respectively.
\end{theorem}

Theorem~\ref{thm:pythagorean}, {whose proof is in Appendix \ref{sec:pythagoreanproof}}, establishes that the $\alpha$-divergence between~$\pi$ and any approximation $\widetilde\pi_r$ in the form of~\eqref{eq:nonlinpitilde} must be a weighted sum of two terms: the divergence between the true distribution and~$\pi_{\alpha,r}^\opt$, and the divergence between~$\pi_{\alpha,r}^\opt$ and the proposed approximation. Since all quantities are positive, we deduce that the function $\ell_{\alpha,r}^\opt$ is the minimizer of $\widetilde\ell_r\mapsto D_\alpha( \pi || \widetilde\pi_r )$ for any feature map $\varphi_r$. The uniqueness of this optimizer also follows Theorem~\ref{thm:pythagorean}: for any $\widetilde\pi_r$ such that $D_\alpha( \pi || \widetilde\pi_r ) = D_\alpha( \pi || \pi_{\alpha,r}^\opt)$, \eqref{eq:pythagorean} implies $D_\alpha( \pi_{\alpha,r}^\opt \,||\, \widetilde\pi_r)=0$ so that $\widetilde\pi_r = \pi_{\alpha,r}^\opt$ is unique.

As mentioned in the introduction, we are particularly interested in the range $0<\alpha <1$, for which we have
\begin{equation}
 \label{eq:optdivergence}
D_\alpha \lb \pi \,||\, \pi_{\alpha,r}^\opt\, \rb
 = \dfrac{1}{\alpha(\alpha-1)}\lb \lb \dfrac{Z_{\alpha,r}}{Z_\pi}\rb^\alpha - 1 \rb,
 \quad Z_{\alpha,r} = \int \lb \int \ell^\alpha\d\mu_{\perp|r} \rb^{\frac{1}{\alpha}} \d\mu_r.
\end{equation}
This demonstrates that the normalizing constant~$Z_{\alpha,r}=Z_{\alpha,r}(\varphi_r)$ essentially characterizes the divergence loss attained by~$\pi_{\alpha,r}^\opt$. {As $0 \leq Z_{\alpha,r} \leq Z_\pi$ by Jensen's inequality, this result illustrates that the optimal feature map~$\varphi_r$ induces a normalizing constant that is as close to $Z_\pi$ as possible.}

We conclude this section with an interesting property of the optimal profile for the (forward) KL divergence ($\alpha = 1$). In this case, we can write $\frac{1}{Z_\pi}\ell_{1,r}^\opt\d\mu_r = \d\pi_r$, where $\pi_r$ is the pushforward measure of $\pi$ by $\varphi_r$.
This decomposes the optimal approximating measure $\pi_{1,r}^\opt$ as
\begin{equation}\label{eq:KLoptimalDecomp}
 \d\pi_{1,r}^\opt(x) = \d\pi_r(\theta_r) \d\mu_{\perp|r}(x|\theta_r).
\end{equation}
Compared to the factorization
$
 \d\pi(x) = \d\pi_r(\theta_r) \d\pi_{\perp|r}(x|\theta_r),
$
we see that the optimal KL approximation $\pi_{1,r}^\opt$ essentially replaces the conditional target measure $\d\pi_{\perp|r}(\cdot|\theta_r)$ with the conditional reference measure $\d\mu_{\perp|r}(\cdot|\theta_r)$.
This nice interpretation is no longer possible for~$\pi_{\alpha,r}^\opt$ with $\alpha\neq1$.
The next proposition shows that $\pi_{1,r}^\opt$ is actually \emph{quasi-optimal} with respect to the $\alpha$-divergence for ${0<\alpha<1}$, in that using the easily interpretable $\pi_{1,r}^\opt$ instead of $\pi_{\alpha,r}^\opt$ increases the approximation error by at most a factor of $1/\alpha$.
{The proof is left to Appendix \ref{proof:KLisQuasiOptimal}.}

\begin{proposition}\label{prop:KLisQuasiOptimal}
 For any $0<\alpha\leq1$ and for any feature map $\varphi_r$, the probability measures $\pi_{\alpha,r}^\opt$ and $\pi_{1,r}^\opt$ defined in \eqref{eq:piopt} and \eqref{eq:opt_llhd} satisfy
 \begin{equation}\label{eq:KLisQuasiOptimal}
  D_\alpha \lb \pi || \pi_{\alpha,r}^\opt \rb
  \leq
  D_\alpha \lb \pi || \pi_{1,r}^\opt \rb
  \leq
  \frac{1}{\alpha} D_\alpha \lb \pi || \pi_{\alpha,r}^\opt \rb .
 \end{equation}
\end{proposition}

\section{Certifiable bound for linear feature maps}
\label{sec:alphaCDR}

In this section we consider only linear feature maps $\varphi_r(x)=U_r^\trans x$, where $U_r\in\R^{d\times r}$ is a matrix with orthonormal columns.
Using the appropriate functional inequalities, which we introduce next, we derive a bound for $D_\alpha(\pi || \pi_{\alpha,r}^\opt)$ involving a simple function of the subspace $U_r$.
{Notably, minimizing our bound is \emph{computationally tractable}. While this construction of $U_r$ may not be the globally optimal solution minimizing the $\alpha$-divergence, its main benefit is that the corresponding approximation error $D_\alpha(\pi || \pi_{\alpha,r}^\opt)$ furnishes a computable certificate of optimality that can be driven to zero as $r \to d$.
}

\subsection{Functional inequalities}\label{sec:FunctionalInequality}
Our main technical tool for bounding $D_\alpha(\pi || \pi_{\alpha,r}^\opt)$ relies on the $\phi$-Sobolev inequality.

\begin{definition}[$\phi$-Sobolev inequality~\cite{Chafai_2004, Bolley_Gentil_2010}]\label{def:phiSI}
Let $\phi : \mathbb{R}_{\geq 0} \to \mathbb{R}_{\geq 0}$ a smooth convex function such that $-\nicefrac{1}{\phi''}$ is also convex. For a given probability measure $\mu$ on $\R^d$, we denote by $C_\phi(\mu)>0$
the smallest constant such that the \emph{$\phi$-Sobolev inequality}
\begin{equation}
\label{eq:phiSI}
\Ent^\phi_{\mu}(f) \leq \frac{C_\phi(\mu)}{2} \,\mathbb{E}_{X\sim\mu} [ \phi''(f(X))\,\| \nabla f(X) \|_2^2 ]
\end{equation}
holds for all sufficiently smooth positive functions~$f: \mathbb{R}^d \to \mathbb{R}_{\geq 0}$, where
\begin{equation}
\label{eq:phientropy}
\Ent_\mu^\phi(f) = \int \phi(f) \d\mu - \phi\left( \int f \d\mu \right)
\end{equation}
denotes the $\phi$-entropy of $f$ under $\mu$ and $\|\cdot\|_2$ the Euclidean norm of $\R^d$.
If~$C_\phi(\mu)<\infty$ we say that $\mu$ satisfies the $\phi$-Sobolev inequality with constant $C_\phi(\mu)$.
\end{definition}

We momentarily postpone a discussion of sufficient conditions for a measure to satisfy a $\phi$-Sobolev inequality. Instead, we note that the significance of the $\phi$-Sobolev inequality to the present work stems from its specialization to the Amari divergence function $\phi=\phi_\beta$, as in~\eqref{eq:amari}, with which \eqref{eq:phiSI} becomes
\begin{equation}\label{eq:betaSI}
\frac{1}{\beta(\beta-1)} \lb \int f^{\beta}\d\mu - \lb \int f \d\mu\rb^{\beta} \rb \leq \frac{C_{\beta}(\mu)}{2}\, \mathbb{E}_{X \sim \mu}[f^{\beta} \|\nabla \ln f(X)\|_2^2 ].
\end{equation}
We refer to this as the \emph{$\beta$-Sobolev inequality} and, to simplify our notation, we write $C_{\beta}(\mu)=C_{\phi_\beta}(\mu)$. The convexity conditions stipulated by $\phi$-Sobolev inequality restrict the value of $\beta$ to the interval $\beta \in [1,2]$.
As pointed out in \cite{Chafai_2004, Bolley_Gentil_2010}, we have that:
\begin{itemize}
\item when $\beta\rightarrow1$, the left-hand side of \eqref{eq:betaSI} becomes the canonical Shannon entropy, and the resulting inequality corresponds to the \emph{logarithmic Sobolev inequality} (LSI) originally proposed by Gross \cite{Gross_1975};
\item when $\beta=2$, the left-hand side of \eqref{eq:betaSI} becomes half the variance of $f$, and the resulting inequality is more commonly referred to as the \emph{Poincar\'{e} inequality}.
\end{itemize}
From this perspective, the family of inequalities \eqref{eq:betaSI} for $\beta \in [1,2]$ can be viewed as an interpolation between the LSI and the Poincar\'e inequality.
We discuss possible modifications beyond $\beta \geq 2$ in Section~\ref{sec:monotone_ext}, and potential improvements for~$\beta\in[1,2]$ in Section~\ref{sec:improvedSI}.

As shown by the following proposition, a direct application of inequality \eqref{eq:betaSI} with $\beta=1/\alpha$ permits us to bound the divergence $D_\alpha(\pi||\mu)$ using the (trace of the) diagnostic matrix~$H$, as in \eqref{eq:defH}.

\begin{proposition}\label{prop:BoundDalphaNoDR}
 Let $\pi$ and $\mu$ be two probability measures on $\R^d$ with $\d\pi(x)\propto \ell(x)\d\mu(x)$ for some smooth $\ell:\mathbb{R}^d\rightarrow\mathbb{R}_{\geq 0}$. Let $\alpha\in[\nicefrac{1}{2},1]$ and assume~$\mu$ satisfies the $\nicefrac{1}{\alpha}$--Sobolev inequality~\eqref{eq:betaSI}. Then
 \begin{equation}\label{eq:BoundDalphaNoDR}
  D_\alpha(\pi\,||\,\mu) \leq
  \mathcal{J}_\alpha\lb C_{1/\alpha}(\mu)\, \mathbb{E}_{X \sim \pi}\left[ \| \nabla \ln \ell(X) \|_2^2 \right]\rb ,
 \end{equation}
 where $\mathcal{J}_\alpha:\R_{\geq0}\rightarrow\R_{\geq0}$ is defined by
 \begin{equation}\label{eq:defJalpha}
  \mathcal{J}_\alpha(t) =
  \frac{1}{\alpha(\alpha-1)} \lb   \lb 1-\frac{(1-\alpha)}{2} t \rb_+^\alpha -1 \rb
 \end{equation}
 for $\alpha\neq 1$ and by $\mathcal{J}_\alpha(t)=\frac{1}{2} t$ for $\alpha=1$, where $(t)_+ = \max\{ t, 0\}$.

\end{proposition}

We refer to $\mathcal{J}_\alpha$ as the \emph{majorized loss function} for dimension reduction with $\alpha$-divergences.
{Observe that this is evaluated with the trace of the diagnostic matrix~$\E_{X \sim \pi}[\| \nabla \ln \ell(X)\|_2^2]$, which is also classically referred to as the \emph{relative Fisher information} of~$\pi$ with respect to~$\mu$.
The relative Fisher information is known to bound the Kullback--Leibler divergence \cite{Gross_1975}, the total variation distance \cite{guillin2009transportation}, and the Hellinger distance \cite{Cui_Tong_2021}. Proposition \ref{prop:BoundDalphaNoDR} extends these results to $\alpha$-divergences with $1/2\leq\alpha\leq1$.
}

Figure~\ref{fig:Jalpha} depicts these loss functions for a range of $\alpha$; the curves associated with $\alpha < \nicefrac{1}{2}$ will be discussed in Section~\ref{sec:monotone_ext}.
We note two significant properties of these functions that are crucial to subsequent developments. First, the image of $\mathcal{J_\alpha}$ never exceeds the vacuous upper bound of $|\alpha(\alpha-1)|^{-1}$ for $\alpha$-divergences; instead, it at most saturates at this plateau. Second, we see by inspection  that $t \mapsto \mathcal{J}_\alpha(t)$ is monotonically non-decreasing  for all $\alpha \in (0,1]$; we show this statement formally in Appendix \ref{sec:comparison}.

\begin{figure}[t!]
\begin{centering}
\includegraphics[width=0.9\textwidth]{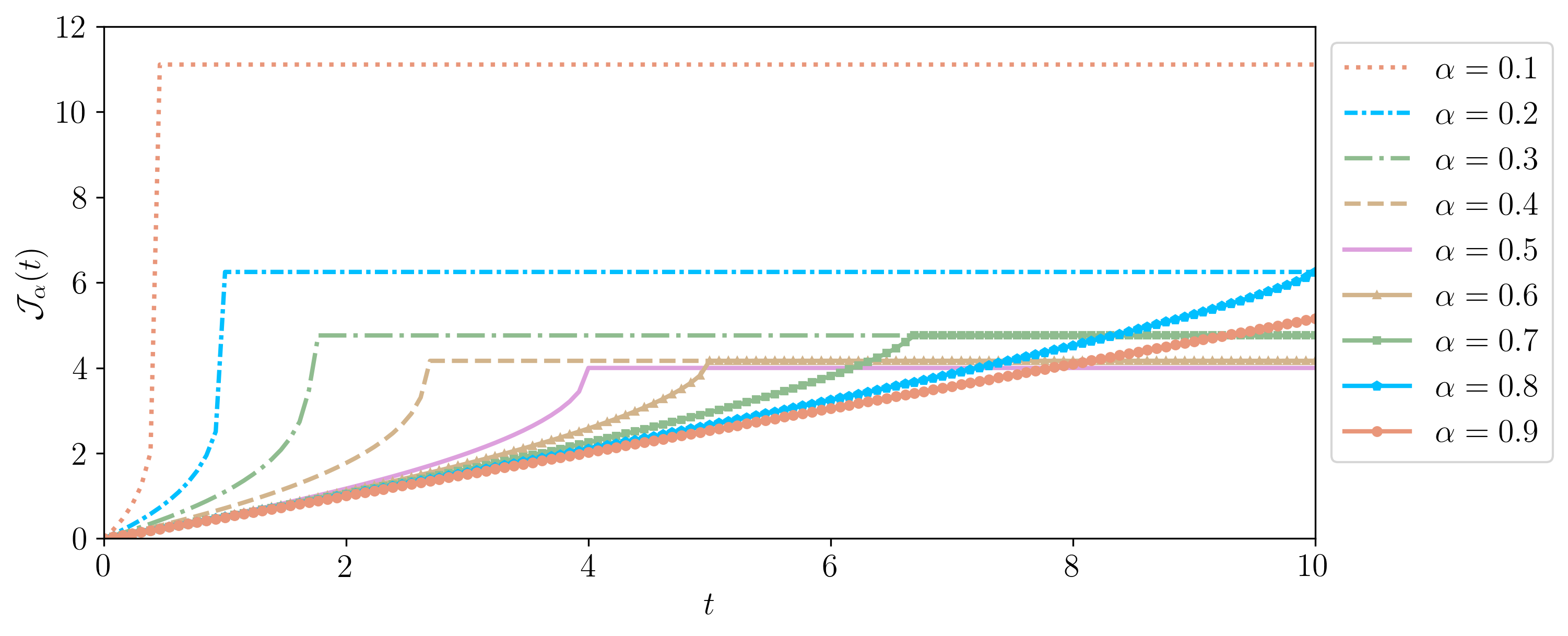}
\end{centering}
\caption{\label{fig:Jalpha} Visualization of the majorized loss function $t \mapsto \mcJ_\alpha(t)$ for $\alpha \geq \nicefrac{1}{2}$, defined in \eqref{eq:defJalpha} (solid lines \solidrule~), and its extension for $0 < \alpha < \nicefrac{1}{2}$, defined in \eqref{eq:Jext} (dashed lines \dashedrule~).}
\end{figure}

\subsection{Derivation of the upper bound for $\nicefrac{1}{2} \leq \alpha \leq 1$}\label{sec:FirstBound}

We now extend the result in Proposition \ref{prop:BoundDalphaNoDR} to bound $D_\alpha(\pi \,||\, \pi_{\alpha,r}^\opt)$. To do so, we need to control the $\phi$-Sobolev constant of conditional measures $\mu_{\perp|r}$ obtained from conditioning $\mu$ on any \emph{linear} feature map. This requirement motivates the introduction of the subspace $\phi$-Sobolev inequality.

\begin{definition}
[Subspace $\phi$-Sobolev Inequality]
A measure $\mu$ on $\R^d$ satisfies a \emph{subspace $\phi$-Sobolev inequality} if there exists a constant $C^\mathrm{sub}_\phi( \mu )<\infty$ such that, for any $0 \leq r\leq d$,
and for any matrix $U_r\in\R^{d\times r}$ with $r\leq d$ orthogonal columns,
the measure $\mu_{\perp|r}$ obtained by conditioning $X\sim\mu$ on the event $X \mid U_r^\trans X = \theta_r$ for any $\theta_r \in \R^r$ satisfies
$$
 C_\phi( \mu_{\perp|r} ) \leq C^\mathrm{sub}_\phi( \mu ) .
$$
We call the smallest possible constant for $C_\phi^\mathrm{sub}( \mu )$ the subspace $\phi$-Sobolev constant.
\end{definition}

We emphasize that the \emph{subspace} $\phi$-Sobolev inequality is a strictly stronger assumption than the $\phi$-Sobolev inequality. Indeed, while a $\phi$-Sobolev inequality for $\mu$ implies that each \emph{marginal} distribution~$\mu_{r}$ satisfies $C_\phi(\mu_r)\leq C_\phi(\mu)$ (simply by injecting $f=f_r\circ \varphi_r$ in \eqref{eq:phiSI}), it does not imply the $\phi$-Sobolev inequality for the \emph{conditional} distributions~$\mu_{\perp | r}$.
The following proposition gives sufficient condition for a measure $\mu$ to satisfy $C^\mathrm{sub}_\phi( \mu )<\infty$; its proof can be found in the supplementary materials.

\begin{proposition}\label{prop:CphiSubFinite}
Any probability measure~$\mu$ with density $\d\mu(x)\propto \exp\left (-V(x)-B(x) \right )\d x$, where $x\mapsto V(x)$ is a strictly log-concave function with $\Hess(V) \succeq RI_d$ for $R > 0$, and $x\mapsto B(x)$ is a bounded function, satisfies the subspace $\phi$-Sobolev inequality with constant
\begin{equation}\label{eq:CphiSubFinite}
 C^\mathrm{sub}_\phi( \mu ) \leq \frac{\exp(\sup B - \inf B)}{R}   .
\end{equation}
\end{proposition}

This demonstrates that a sufficiently rich class of reference measures satisfies the subspace $\phi$-Sobolev inequality. For example, the subspace $\phi$-Sobolev constant for the isotropic Gaussian measure is $C^\mathrm{sub}_\phi( \mu ) \leq 1$, {and for uniform measures on compact and convex sets $\Omega$, the subspace $\phi$-Sobolev constant is bounded by $C^\mathrm{sub}_\phi( \mu ) \leq \text{diam}(\Omega)^2\exp(1)/4$; see Examples 2.6 and 2.8 in \cite{ZCLSM22}.
We refer to \cite{ZCLSM22} for additional examples of measures which do (or do not) satisfy the assumptions in Proposition~\ref{prop:CphiSubFinite}.}

The subspace $\phi$-Sobolev inequality provides the key ingredient to deriving our error bound. Similar to our earlier proposition, we note that to control dimension reduction with $\alpha$-divergences we must impose the $\beta$-Sobolev inequality of order $\beta = \nicefrac{1}{\alpha}$. As we have so far only established these inequalities for $\beta \in [1,2]$, accordingly this limits the $\alpha$-divergences we consider to the interval $[\nicefrac{1}{2}, 1]$. Notably, this precludes us from considering any of the dual divergences to this interval, which are related by $1-\alpha$.

\begin{theorem}
\label{thm:alphaCDR}
Let $\pi$ and $\mu$ be two probability measures on $\R^d$ with $\d\pi(x) \propto \ell(x)\d\mu(x)$ for some smooth $\ell:\mathbb{R}^d\rightarrow\mathbb{R}_{\geq 0}$.
Then, for all $\nicefrac{1}{2}\leq\alpha\leq1$ and for any matrix $U_r\in\R^{d\times r}$ with $r\leq d$ orthogonal columns, the measure $\pi_{\alpha,r}^\opt(x)$ as in \eqref{eq:piopt} with $\varphi_r(x)=U_r^\trans x$ satisfies
\begin{equation}\label{eq:alphaCDR}
D_\alpha(\pi \,||\, \pi_{\alpha,r}^\opt) \leq
\mcJ_\alpha\lb C^\mathrm{sub}_{1/\alpha}( \mu )\,
\mathbb{E}_{X\sim\pi} \left[ \|U_\perp^\trans \nabla \ln \ell(X) \|_2^2 \right]
\rb,
\end{equation}
where $U_\perp\in\R^{d\times(d-r)}$ is any orthogonal completion of $U_r$
and where $\mcJ_\alpha$ is the function defined in \eqref{eq:defJalpha}.

\end{theorem}

We now consider minimizing this error bound.
Since $t \mapsto \mcJ_\alpha(t)$ is monotone non-decreasing, we have
\begin{align}
& \argmin_{U_r \in \,\mathbb{R}^{d \times r}} \mcJ_\alpha\lb C^\mathrm{sub}_{1/\alpha}( \mu ) \mathbb{E}_{X\sim\pi} \left[ \|U_\perp^\trans \nabla_x \ln \ell(X) \|_2^2 \right]\rb \nonumber\\
&=\argmin_{U_r \in \,\mathbb{R}^{d \times r}} \mathbb{E}_{X\sim\pi} \left[ \|U_\perp^\trans \nabla_x \ln \ell(X) \|_2^2 \right] \nonumber\\
&=\argmax_{U_r \in \,\mathbb{R}^{d \times r}} \trace(U_r^\trans H U_r), \label{eq:min_monotone}
\end{align}
where~$H$ is defined in \eqref{eq:defH}.
The optimal solution to~\eqref{eq:min_monotone} is a classical result from~\cite{Fan_1949}, and it follows that columns of the optimal $U_r$ are the $r$~leading eigenvectors of the diagnostic matrix. This result extends the conclusions of~\cite{ZCLSM22} and~\cite{Cui_Tong_2021} by demonstrating that this subspace is \emph{universally robust} for all $\alpha$-divergences interpolating between LSI and the Poincar\'e inequality.
Explicitly, denoting by $\lambda_k$ the $k$-th largest eigenvalue of~$H$, the resulting feature map $\varphi_r(x)=U_r^\trans x$ yields
\begin{equation}
\label{eq:boundloss}
D_\alpha\lb \pi || \pi^\opt_{\alpha,r} \rb
\leq
\mcJ_\alpha\lb C^\mathrm{sub}_{1/\alpha}( \mu ) \sum_{k = r+1}^d \lambda_k \rb .
\end{equation}
This bounds allow the determination of the the reduced dimension (depending on error tolerance and computational budget) \emph{without} re-solving the optimization for each $r$. This is in contrast to other variational dimension reduction algorithms.

This bound also highlights the interplay between the choice of $\alpha$, the subspace Sobolev constant $C^\sub_{1/\alpha}(\mu)$, and the eigenvalues $\lambda_k$ of the diagnostic matrix. We note that if $\sum_{k > r}\lambda_k > \frac{2}{C_{1/\alpha}^\sub(1-\alpha)}$, then the corresponding upper bound achieves the vacuous limit $|\alpha(\alpha-1)|^{-1}$. In other words, if the spectrum of the diagnostic matrix does not decay sufficiently quickly, then the majorized loss function will not provide an informative bound unless sufficiently large $r$ is considered. While this is not an obstruction for the KL divergence ($\alpha=1$), this may be a bottleneck when considering dimension reduction for, e.g., the squared Hellinger distance ($\alpha = \nicefrac{1}{2}$). In general, the lower the value of $\alpha$, or the smaller the constant $C^\sub_{1/\alpha}(\mu)$ of the reference measure, the more stringent this requirement becomes; see e.g., Figure~\ref{fig:alphacomparision} for a visualization.

\begin{figure}[t!]
\begin{centering}
\includegraphics[width=0.9\textwidth]{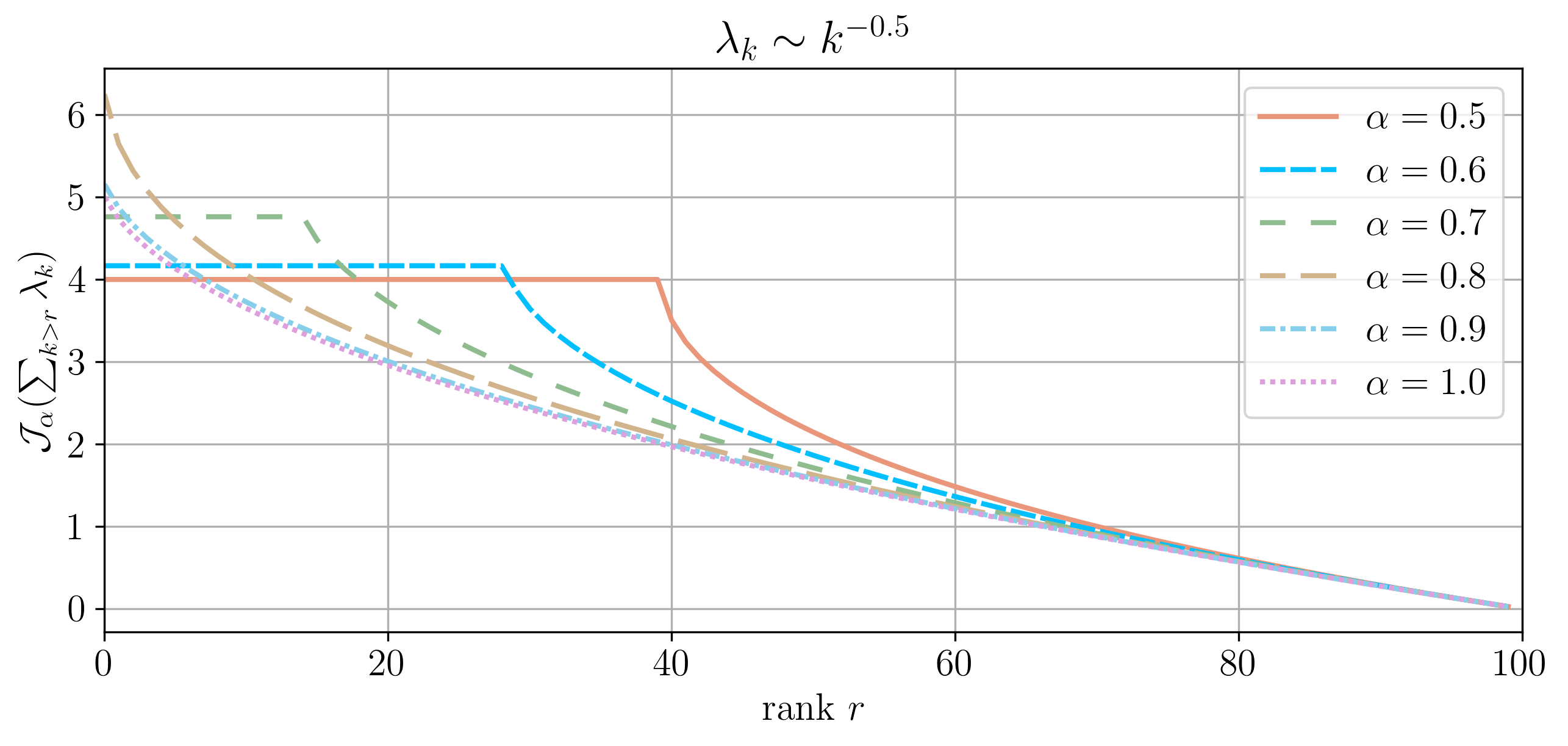}
\end{centering}
\caption{\label{fig:alphacomparision} Comparison of the majorization~\eqref{eq:boundloss} across different $\alpha \in (\nicefrac{1}{2},1]$. The decay of the eigenvalue spectrum is assumed to be algebraic, and the trace normalization of the diagnostic matrix is assumed to be $\sum \lambda_k = 10$ for this example.
}
\end{figure}

\begin{remark}
A similar bound for the squared Hellinger metric was obtained in \cite{Cui_Tong_2021} which reads $d_\textsc{Hell}^2(\pi, \pi^\opt_{\nicefrac{1}{2},r}) \leq \frac{1}{4}C^\sub_2 \sum_{k > r}\lambda_k$, where $\lambda_k$ denotes the $k$-th eigenvalue of the diagnostic matrix. By Theorem~\ref{thm:alphaCDR}, we improve on this bound since
\[
d_\textsc{Hell}^2(\pi, \pi^\opt_{\nicefrac{1}{2},r}) \leq \frac{1}{4} \mathcal{J}_{1/2}\lb C^\sub_2 \sum_{k > r}\lambda_k\rb \leq \frac{1}{4}C^\sub_2 \sum_{k > r}\lambda_k,
\]
which can be verified by applying Bernoulli's inequality. Moreover, in the limit $\sum_{k>r}^d \lambda_k \to 0$ as $r \to d$, our approximation becomes sharper than theirs by a factor of two. A visual comparison of these bounds is provided by Figure~\ref{fig:comparisonCuiTong} for a synthetic linear-Gaussian example (see Appendix \ref{sec:lingaussian} for details).
\end{remark}

\begin{figure}[h!]
\begin{centering}
\includegraphics[width=0.9\textwidth]{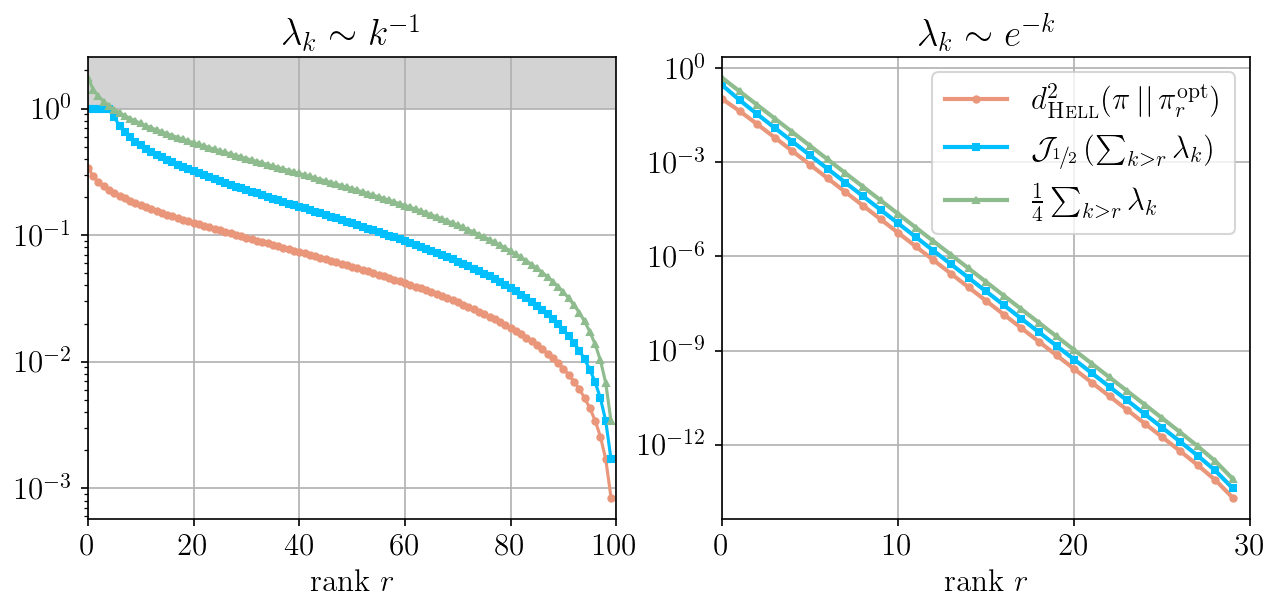}
\end{centering}
\caption{ Comparison of the exact squared Hellinger loss for the linear Gaussian inverse problem (Appendix \ref{sec:lingaussian}), the majorized bound in $\mcJ_{\nicefrac{1}{2}}(\sum_{k>r} \lambda_k)$ in~\eqref{eq:boundloss}, and the bound $\frac{1}{4} \sum_{k>r}\lambda_k$ derived by Cui and Tong \cite{Cui_Tong_2021}. (Left) Example with algebraically decaying eigenvalues of the diagnostic matrix with $d=100$ and normalization $\sum_{k=1}^d \lambda_k = 7$. The shaded region indicates a vacuous upper-bound. (Right) Example with exponentially decaying eigenvalues for $d = 50$ and normalization $\sum_{k=1}^d \lambda_k = 2$.
}
\label{fig:comparisonCuiTong}
\end{figure}

\subsection{Extension of upper bound to $0 < \alpha \leq 1$}
\label{sec:monotone_ext}

The upper bounds obtained in Theorem~\ref{thm:alphaCDR} are valid only for the interval $\alpha \in [\nicefrac{1}{2},1]$, essentially because the convexity requirements of the $\phi$-Sobolev inequality restrict us to considering $\beta$-Sobolev inequalities for $\beta=\nicefrac{1}{\alpha} \in [1,2]$. It is natural to wonder whether this is a fundamental obstruction, or perhaps merely an artificial restriction imposed by the proof technique of the inequalities.
Bolley and Gentil~\cite{Bolley_Gentil_2010} address this question and note the regime $\beta > 2$ admits a trivial extension. Their key insight is to manipulate the $\beta$-Sobolev inequalities into the form of \emph{Beckner} inequalities~\cite{Beckner_1989} by choosing the test function $f = g^{\nicefrac{2}{\beta}}$ so that \eqref{eq:betaSI} yields
\begin{equation}\label{eq:Beckner}
\frac{1}{(\beta-1)} \lb \int g^2\d\mu - \lb \int g^{\frac{2}{\beta}} \d\mu\rb^{\beta} \rb \leq \frac{2 C_{\beta}(\mu)}{\beta}\, \mathbb{E}_{X \sim \mu}[ g(X)^2 \,\|\nabla \ln g(X)\|_2^2 ].
\end{equation}
This form is then amenable to the use of the following monotonicity relation that originally appeared in \cite[Lemma 1]{Latala_Oleszkiewicz_2000}.
\begin{lemma}[Beckner monotonicity]
\label{lem:beckner_monotone}
For any probability measure~$\mu$ on $\mathbb{R}^d$ and positive-valued function $g : \mathbb{R}^d \to \mathbb{R}$, the mapping
\[
\beta \mapsto \frac{1}{(\beta-1)} \lb \int g^2\d\mu - \lb \int g^{\frac{2}{\beta}} \d\mu\rb^{\beta} \rb
\]
is non-increasing for all $1 < \beta < \infty$.
\end{lemma}

By applying this lemma and reverting to $\beta$-Sobolev form using the inverse transformation $g = f^{\nicefrac{\beta}{2}}$, we recover the following statement: if a measure~$\mu$ satisfies the Poincar\'e inequality ($\beta=2$) with constant $C_2(\mu)$, then it also satisfies a $\beta$-Sobolev inequality for $\beta \geq 2$ with constant $\frac{\beta}{2}C(\mu)$. (In other words, we have the bound $C_\beta(\mu)\leq \frac{\beta}{2} C_2(\mu)$ for all $\beta\geq2$, although there is no guarantee that such an estimate is sharp.)
It immediately follows that
$$
 C_\beta^\textrm{sub}(\mu)\leq \frac{\beta}{2} C_2^\textrm{sub}(\mu), \quad \forall \beta\geq2.
$$
This is sufficient to enable us to extend our majorization for all $\alpha \in (0,1]$.

\begin{theorem}
\label{thm:alphaCDRext}
Let $\pi$ and $\mu$ be two probability measures on $\R^d$ with $\d\pi(x) \propto \ell(x)\d\mu(x)$ for some smooth $\ell:\mathbb{R}^d\rightarrow\mathbb{R}_{\geq 0}$.
Then, for any $0<\alpha\leq1$ and for any matrix $U_r\in\R^{d\times r}$ with $r\leq d$ orthogonal columns, the measure $\pi_{\alpha,r}^\opt(x)$ as in \eqref{eq:piopt} with $\varphi_r(x)=U_r^\trans x$ satisfies
\begin{equation}\label{eq:alphaCDRext}
D_\alpha(\pi || \pi_{\alpha,r}^\opt) \leq \mcJ_\alpha \lb C_{\min\{\frac{1}{\alpha};\,2\}}^\sub(\mu)\, \E_{X \sim \pi}\left[ \|U_\perp^\trans \nabla \ln \ell (X)\|_2^2 \right] \rb,
\end{equation}
where $U_\perp\in\R^{d\times(d-r)}$ is any orthogonal completion of $U_r$
and where
\begin{equation}
\label{eq:Jext}
\mcJ_\alpha(u) =
\begin{cases}
\dfrac{1}{\alpha(\alpha-1)} \left(\lb 1 - \dfrac{(1-\alpha)}{4 \alpha} u \rb_+^\alpha - 1 \right)  & \alpha \in (0,\,\nicefrac{1}{2}], \\[10pt]
\dfrac{1}{\alpha(\alpha-1)} \left( \lb 1 - \dfrac{(1-\alpha)}{2} u \rb_+^\alpha - 1 \right) & \alpha \in (\nicefrac{1}{2},\,1].
\end{cases}
\end{equation}
\end{theorem}

It is straightforward to check that~\eqref{eq:Jext} remains monotone non-decreasing even in the extended regime $\alpha < \nicefrac{1}{2}$; see also Figure~\ref{fig:Jalpha} for a visualization. As such, the equivalence~\eqref{eq:min_monotone} persists and therefore optimization for $U_r$ proceeds as previously discussed.
Moreover, provided the spectrum of the diagnostic matrix satisfies $\sum_{k>r} \lambda_k < \frac{4\alpha}{C^\sub_{\min\{\nicefrac{1}{\alpha};2\} }(1-\alpha)}$,
then the duality property $D_\alpha(\nu||\mu) = D_{1-\alpha}(\mu||\nu)$ of the $\alpha$-divergences allows us to certify dimension reduction with respect to the \emph{forward} and \emph{reverse} divergences simultaneously. (Even still, this precludes us from considering the \emph{reverse} KL divergence since this eigenvalue sum-bound goes to zero as $\alpha \to 0$.)
Remarkably, even in this setting we observe that the $U_r$ constructed from the eigenvectors of the diagnostic matrix continues to provide a certificate of optimality.

\begin{remark}
The monotonicity result of Lemma~\ref{lem:beckner_monotone} ensures that a measure satisfying a $\beta$-Sobolev inequality with $C_\beta(\mu)<\infty$ also satisfies a $\beta'$-Sobolev inequality for any $\beta'\geq \beta$.  {As \cite[\S 7.6.2]{Bakry2014} notes, however, an application of Jensen's inequality shows that the reverse implication is also true for all inequalities \textit{except} for $\beta = 1$. Accordingly, the $\beta$-Sobolev inequalities for $\beta \in (1,\infty)$ are in fact equivalent to the Poincar\'e inequality. Meanwhile, the log-Sobolev inequality ($\beta=1$) is exceptional since
the so-called Herbst argument demonstrates that only sub-Gaussian measures $\mu$ can satisfy this inequality~\cite{Bakry2014}.  This makes it sufficiently distinct from the Poincar\'e inequality since, while the measure $\d\mu \propto \exp(-\|x\|_1)\d x$ satisfies the latter, it evidently cannot satisfy the LSI on account of its slow tail decay. One of the anonymous reviewers brought to our attention that an alternative family of functional inequalities which genuinely interpolate between the LSI and the Poincar\'e inequality, in this sense of \emph{tail conditions}, are given by the Lata\l{}a--Oleskiewicz (LO) inequalities~\cite{Latala_Oleszkiewicz_2000} (see also \cite[\S 7.6.3]{Bakry2014}).
}
\end{remark}

\subsection{Extension to other $\phi$-divergences and distances}
\label{sec:extensionother}

The preceding discussion provides certifiable bounds for $\alpha$-divergences with $\alpha \in (0,1]$, but does not address whether further extensions beyond this interval are possible by using other modifications to the functional inequalities. Indeed, Bolley and Gentil also considered whether the $\beta$-Sobolev inequalities can be extended for $\beta < 1$. Unfortunately, they provide a simple counter-example for the interval $\beta \in (0,1)$ which indicates that inequalities of the form~\eqref{eq:betaSI} cannot hold for the Gaussian measure \cite[p.465]{Bolley_Gentil_2010}. Adapted to our context, this suggests that $\alpha$-divergences corresponding to $\alpha \in (1,\infty)$ cannot be majorized using this approach; notably, this includes the $\chi^2$-divergence.

We note that beyond the $\alpha$-divergences considered here, there are many alternative $\phi$-divergences which would be of interest in the context of dimension reduction. One strategy we propose is to leverage known inequalities which relate generic $\phi$-divergences to $\alpha$-divergences. For example, the total variation (TV) distance
commonly appears in statistical applications, and we can express
\[
\textsc{TV}(\pi\,||\,\pi^\opt_{TV,r}) \leq \textsc{TV}(\pi\,||\,\pi^\opt_{1,r}) \leq \sqrt{ \textsc{KL}(\pi, || \, \pi^\opt_{1,r})},
\]
where the first inequality follows by definition of the optimal profile function for the total variation divergence, which we have not characterized, and the second relation is commonly referred to as Pinsker's inequality.\footnote{One could alternatively invoke the globally non-vacuous (but asymptotically weaker) bound $\textsc{TV}(\pi, \nu) \leq \sqrt{1 - \exp(-\textsc{KL}(\pi||\nu))}$; see \cite{Canonne_2022} for details.} It immediately follows from Theorem~\ref{thm:alphaCDR} that our certificate of optimality for the KL divergence concurrently certifies
$
\textsc{TV}(\pi \, || \, \pi^\opt_{TV, r}) \leq \sqrt{\frac{1}{2} C^\sub_1(\mu) \sum_{k > r}^d \lambda_k }\,,
$
although with no understanding of when such inequalities are sharp. More generally, the joint-range result of \cite{Harremoes_Vajda_2011} characterizes when one $\phi$-divergence can bound another, and so in particular one can relate any $\phi$-divergence to the $\alpha$-divergences for $\alpha \in (0,1]$ (provided such an upper bound exists). It is unclear how to choose the $\alpha$ which provides the tightest certificate of optimality; for the TV divergence we conjecture that one should choose the $\alpha=\nicefrac{1}{2}$ divergence, although we have not done this computation.

\begin{remark}
The joint-range inequalities also provide comparisons between different $\alpha$-divergences. For example, a common joint-range relation between the KL divergence and the squared-Hellinger distance is given by $d^2_\textsc{Hell}(\nu, \mu) \leq \frac{1}{2}\textsc{KL}[\nu || \mu]$, which implies  that the KL certificate of optimality also concurrently certifies dimension reduction for the squared Hellinger distance. In fact, this recovers the bound of \cite{Cui_Tong_2021}, although we note this is suboptimal compared to the bounds derived directly from the functional inequalities as we did here.
\end{remark}

\section{An improved bound for linear feature maps}
\label{sec:generalimprove}

The previous section establishes the connection between the (subspace) $\beta$-Sobolev inequalities and dimension reduction for $\alpha$-divergences. Using this roadmap, we can benefit from recent improvements to the $\beta$-Sobolev inequalities by Bolley and Gentil \cite{Bolley_Gentil_2010} to obtain tighter majorized loss functions for each approximation~$\pi^\opt_{\alpha,r}$ for all $\alpha \in (\nicefrac{1}{2},1)$.
For ease of presentation we adopt the condensed notation
$
 \mu(f) = \int f \d\mu
$
for the expectation.

\subsection{An improved $\beta$-Sobolev inequality for~$\beta \in (1,2)$}
\label{sec:improvedSI}

Bolley and Gentil~ \cite[Corollary~10]{Bolley_Gentil_2010} provides the following refined bounds for the $\beta$-Sobolev inequality for~$\beta \in (1,2)$.
\begin{lemma}[Improved $\beta$-Sobolev Inequality]
\label{lem:betaSIimproved}
Suppose a measure~$\mu$ satisfies the $\beta$-Sobolev inequality with constant~$C_\beta(\mu)$ for some $\beta \in (1,2)$. Then, for that same $\beta$, it also satisfies the \emph{improved $\beta$-Sobolev inequality}
\begin{equation}
\label{eq:betaSIimproved}
\frac{1}{2(\beta-1)^2} \lb \mu(f^\beta) - \mu(f)^{2\beta -2} \mu(f^\beta)^{\frac{2}{\beta}-1} \rb \leq \frac{C_\beta(\mu)}{2}\,\mu(f^\beta \, \|\nabla \ln f \|_2^2)
\end{equation}
for any sufficiently smooth positive functions~$f:\mathbb{R}^d \rightarrow \mathbb{R}_{\geq 0}$.
\end{lemma}

This improves upon the traditional $\beta$-Sobolev inequality since the left-hand side of~\eqref{eq:betaSIimproved} dominates the left-hand side of~\eqref{eq:betaSI}, which can be verified by applying the weighted arithmetic-mean geometric-mean inequality. Moreover, note that these inequalities continue to interpolate the Poincar\'{e} inequality and the log-Sobolev inequality at the endpoints of the interval.

Bolley and Gentil~\cite[Proposition 11]{Bolley_Gentil_2010} also prove a monotonicity result analogous to Lemma~\ref{lem:beckner_monotone}, which states:
\begin{lemma} For any probability measure~$\mu$ on $\R^d$ and a positive valued function $g: \R^d \to \R$, the mapping
\[
\beta \mapsto \frac{\beta}{2(\beta-1)^2} \lb \mu(g^2) - \mu(g^\frac{2}{p})^p \lb \frac{\mu(g^2)}{\mu(g^\frac{2}{p})^p} \rb^{\frac{2}{p}-1} \rb
\]
is non-increasing for all $1 < \beta < \infty$.
\end{lemma}
This lemma relates to~\eqref{eq:betaSIimproved} through the Beckner-type test functions $g = f^{\frac{2}{\beta}}$. As before, the main implication is that if a measure $\mu$ satisfies an improved $\beta$-Sobolev inequality with constant $C_\beta(\mu)$, then it also satisfies the inequality for any $\beta' > \beta$ provided we adapt the constant to $C_{\beta'}(\mu) \leq \frac{\beta'}{\beta} C_{\beta}$. While this will not provide sharp constants in general, its main utility is to extend the results of Lemma~\ref{lem:betaSIimproved} beyond $\beta > 2$, provided we assume our reference measure satisfies the Poincar\'e inequality with $\beta=2$.

\subsection{Derivation of improved upper bound for $\nicefrac{1}{2} < \alpha < 1$}
We leverage this new functional inequality to derive an alternative majorization for dimension reduction with $\alpha$-divergences.
\begin{theorem}
\label{thm:alphaCDRimp}
Let $\pi$ and $\mu$ be two probability measures on $\R^d$ with $\d\pi(x) \propto \ell(x)\d\mu(x)$ for some smooth $\ell:\mathbb{R}^d\rightarrow\mathbb{R}_{\geq 0}$.
Then, for any $0<\alpha\leq1$ and for any matrix $U_r\in\R^{d\times r}$ with $r\leq d$ orthogonal columns, the measure $\pi_{\alpha,r}^\opt(x)$ as in \eqref{eq:piopt} with $\varphi_r(x)=U_r^\trans x$ satisfies
\begin{equation}
\label{eq:alphaCDRimp}
D_\alpha(\pi \,||\, \pi^\opt_{\alpha,r}) \leq \mcJ_\alpha^\flat \lb C_{\min\{\frac{1}{\alpha};\,2\}}^\sub(\mu) \,\mathbb{E}_\pi\left[ \|U_\perp^\trans \nabla_x \ln \ell \|_2^2 \right]\rb ,
\end{equation}
where $U_\perp\in\R^{d\times(d-r)}$ is any orthogonal completion of $U_r$ and where
\begin{equation}
\label{eq:Jflat}
\mcJ_\alpha^\flat(t) =
\begin{cases}
\dfrac{1}{\alpha(\alpha-1)} \left[ \lb 1 - \dfrac{(1-\alpha)^2}{2\alpha} t \rb_+^\frac{\alpha}{2(1-\alpha)} - 1 \right]
& \alpha \in (0, \nicefrac{1}{2}], \\[10pt]
\dfrac{1}{\alpha(\alpha-1)} \left[ \lb 1 - (1-\alpha)^2 t \rb_+^\frac{\alpha}{2(1-\alpha)} - 1 \right]
& \alpha \in (\nicefrac{1}{2}, 1].
\end{cases}
\end{equation}
\end{theorem}

We refer to $\mcJ_\alpha^\flat$ as the \emph{improved majorized loss function} for dimension reduction with $\alpha$-divergences for $\alpha \in (\nicefrac{1}{2}, 1)$. Although Lemma~\ref{lem:betaSIimproved} provides a uniform improvement over the $\beta$-Sobolev inequality, in deriving $\mcJ_\alpha^\flat$ we coarsened the inequality by applying H\"older's inequality. Nevertheless, we show in the Appendix \ref{sec:comparison} that $\mcJ_\alpha^\flat(t) \leq \mcJ_\alpha(t)$ uniformly for all $\alpha \in [\nicefrac{1}{2},1]$; see also Figure~\ref{fig:Jalpha_vs_J}. Note, however, in contrast we have $\mcJ_\alpha(t) < \mcJ_\alpha^\flat(t)$ for $0 < \alpha < \nicefrac{1}{2}$. Both bounds coincide for $\alpha = \nicefrac{1}{2}$, and plotting these curves suggests that
$\lim_{\alpha \to 0}\mcJ^\flat_\alpha(t) = \mcJ_1(t)$, but we have not proven this equivalence formally (this result is expected since the improved $\beta$-Sobolev inequalities interpolate the log-Sobolev inequality, which was used to derive $\mcJ_1$).

\begin{figure}[h!]
\begin{centering}
\includegraphics[width=0.9\textwidth]{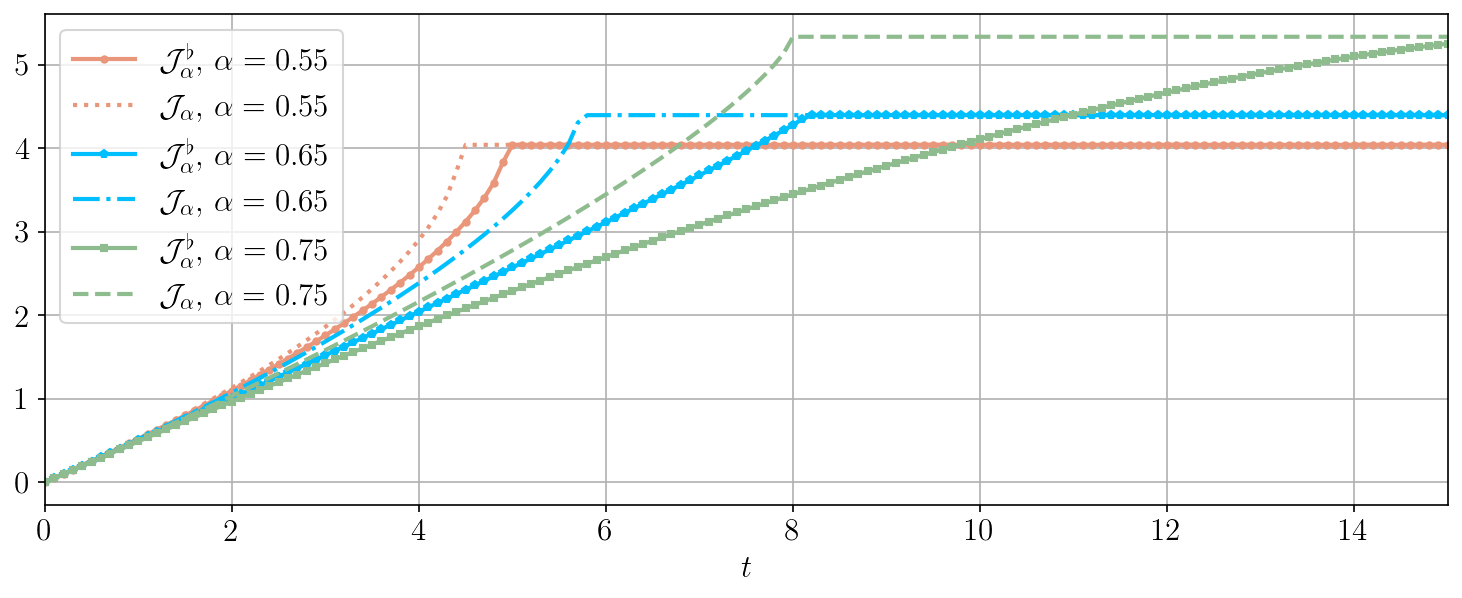}
\end{centering}
\caption{\label{fig:Jalpha_vs_J} Comparison between the improved majorized loss function  $\mcJ^\flat_\alpha(t)$ in~\eqref{eq:Jflat} and the majorized loss function $\mcJ_\alpha(t)$ in~\eqref{eq:defJalpha} for several choices of $\alpha > \nicefrac{1}{2}$.}
\end{figure}

\section{Application to Bayesian inverse problems}
\label{sec:bayesian}

In this section, we specialize our results to the setting of Bayesian inverse problems \cite{stuartActa,KaipioSomersalo}, with the parameter $X$ endowed with a prior $\mu$ on $\R^d$ and the data $Y$ taking values in $\R^m$.
The statistical model for the data is described by a conditional measure $\d\rho^x(y)\propto \ell^y(x) \d y$, where $x\mapsto \ell^y(x)$ is the likelihood function.
Given a realization $y$ of the data, the posterior measure is given by $\d\pi^y(x) \propto \ell^y(x)\d\mu(x)$.
As before, we can consider an approximation $\d\widetilde\pi_r^y(x) \propto \widetilde\ell_r^y \left ((U^y_r)^\trans x \right ) \d\mu(x)$,
which presumes that data are informative only for the reduced parameters~$(U^y_r)^\trans X$. The results from the previous sections can be directly applied in this setting, and yield a posterior approximation $\widetilde\pi_r^y$ where both the feature map $U_r^{y}$ and the associated optimal profile $\ell_{r,\alpha}^{y,\opt}$ depend on $y$.

Alternatively, the idea proposed in \cite{Cui_Zahm_2021} is to seek a feature map which permits one to control the error \emph{averaged} over realizations of the data.
In other words, we consider the following variation of the dimension reduction problem
\begin{equation}
\label{eq:bayesloss}
\min_{\substack{V_r \in \R^{d \times r} \\ V_r^\trans V_r = I_r}} \, \E_{Y\sim\rho} \left[ \, \min_{\widetilde\ell_r^Y:\,\R^r \to \R_{\geq 0}} D_\alpha(\pi^Y \,||\, \widetilde\pi^Y_{r}) \right] \quad \text{for }\d\widetilde\pi_r^y(x) \propto \widetilde\ell_r^y(V_r^\trans x) \d\mu(x).
\end{equation}
The primary difference compared to~\eqref{eq:loss} is that the optimality of~$V_r$ is considered after taking an expectation over the data $Y\sim\rho$, where $\d\rho(y) = \left (\int \rho^x(y)\d\mu(x) \right ) \d y $ is the marginal distribution of the data. Proceeding as before, we can apply Theorem~\ref{thm:pythagorean} since the interior objective $\widetilde\ell^Y_r \mapsto D_\alpha(\pi^Y||\widetilde\pi_r^Y)$ remains identical. Thus, Theorem \ref{thm:alphaCDR} permits to derive the upper bound
\begin{equation}
\label{eq:directavgbound}
\E_{Y\sim\rho} \left[ D_\alpha(\pi^Y \,||\,\widetilde\pi^{Y,\opt}_{\alpha,r}) \right]\leq \E_{Y\sim\rho}\left[ \mathcal{J}_\alpha\lb C_{1/\alpha}^\sub(\mu) \, \E_{X \sim \pi^Y}\| V_r^\trans \nabla_x \ln \ell^Y(x) \|_2^2 \rb \right],
\end{equation}
for any $V_r \in \R^{d \times r}, V_r^\trans V_r = I_r$.
The main complication with~\eqref{eq:directavgbound} is that this majorization is non-trivial to optimize. Although Riemannian gradient methods can be applied, without additional assumptions there are no provable guarantees that such algorithms produce the globally optimal subspace. An exception to this is when $\alpha=1$, as $\E_Y[\mathcal{J}_1(\cdot)] = \mathcal{J}_1(\E_Y [\cdot])$, and minimizing this shows the optimal $V_r$ to be the leading eigenvectors of the \emph{data-free} diagnostic matrix $H_\text{DF} = \E_Y[H]$; this was first noted by~\cite{Cui_Zahm_2021}. While this commutation does not hold in general, one might expect to apply Jensen's inequality to obtain the bound ``$\E_Y \mathcal{J}_\alpha(\cdot) \leq \mathcal{J}_\alpha(\E_Y[\cdot])$''. Unfortunately, we show in the Appendix \ref{sec:comparison} that $\mathcal{J}_\alpha$ in~\eqref{eq:defJalpha} is convex for $\alpha \in (0,1)$, whereas this relation requires \emph{concavity}.

Serendipitously, our improved bound $\mcJ^\flat_\alpha$ is concave for $\alpha \in [\nicefrac{2}{3}, 1]$; see, e.g., Figure~\ref{fig:Jalpha_vs_J}, as well as the Appendix \ref{sec:comparison} for a formal proof. Conversely, when $\alpha < \nicefrac{2}{3}$ we have the crude upper bound $\mcJ_\alpha(t) \leq \frac{t}{2\alpha}$, for $\alpha \in (0,1]$, which is trivially concave. Together, this allows us to pursue the strategy outlined above. We formalize this result with the following theorem.

\begin{theorem}
\label{corr:datafreebound}
Let $\pi^Y$ and $\mu$ be two probability measures on $\R^d$ with $\d\pi^Y(x)\propto \ell^Y(x)\d\mu(x)$ for some smooth integrable function $\ell^Y:\mathbb{R}^d\rightarrow\mathbb{R}_{\geq 0}$ which depends on a random variable $Y$.
Consider the probability measure
\begin{equation}
\d\pi^{Y,\opt}_{\alpha, r} \propto \ell^{Y,\opt}_{\alpha,r}(V_r^\trans x) \,\d\mu(x)
\end{equation}
for any matrix $V_r \in \R^{d \times r}$ with $r \leq d$ orthogonal columns.
If~$\mu$ satisfies a subspace $\beta$-Sobolev inequality, then we have for all $\alpha \in (0,1]$ the inequality
\begin{equation}
\label{eq:boundbayesloss}
\E_{Y\sim\rho} \left[ D_\alpha(\pi^Y || \pi^{Y,\opt}_{\alpha,r}) \right] \leq
\mcJ_\alpha^\text{DF} \lb  C^\sub_{\min\{\frac{1}{\alpha};2\}}(\mu)\,\E_{X \sim \mu} \left[ \E_{Y\sim\rho^X} [\,\|V_\perp^\trans \nabla \ln \ell^Y(X)\|_2^2] \right] \rb,
\end{equation}
where $V_\perp$ is any orthogonal completion to $V_r$, and where
\[
\mcJ^\text{DF}_\alpha (t) =
\begin{cases}
\mcJ^\flat_\alpha(t)~\text{as defined in~\eqref{eq:Jflat}} & \alpha \in [\frac{2}{3},1], \\[5pt]
\frac{1}{2\alpha}t & \alpha \in (0, \frac{2}{3}).
\end{cases}
\]
\end{theorem}

By monotonicity of $\mcJ^\text{DF}_\alpha$ it is straightforward to show (by similar arguments to the preceding sections) that the optimal $V_r$ which minimizes the upper bound in~\eqref{eq:boundbayesloss} is given by any matrix containing the $r$-leading eigenvectors of the
\emph{data-free diagnostic matrix}
\begin{align*}
H_\text{DF} = \E_{X\sim\mu}\left[ \mathcal{I}(X) \right] , \qquad  \mathcal{I}(x) = \E_{Y\sim\rho^x} \left[ \nabla \ln \ell^Y(x) \nabla \ln \ell^Y(x)^\trans \right] ,
\end{align*}
where $\mathcal{I}(x)$ is the Fisher information matrix of the statistical model $\rho^x(\cdot)$ for $Y$,
evaluated at $x$.
Denoting by $\lambda_k(H_\text{DF})$ the $k$-th largest eigenvalue of this data-free diagnostic matrix, we then have the certificate of optimality
\[
\E_{Y\sim\rho} \left[ D_\alpha(\pi^Y \,||\, \pi^{Y,\opt}_{\alpha,r} ) \right] \leq
\mcJ_\alpha^\text{DF}\lb C^\sub_{\min\{\frac{1}{\alpha};2\}}(\mu) \sum_{k > r}^d \lambda_k(H_\text{DF})\rb,
\]
which relates the approximation error, averaged over all realizations of data~$Y$, to
the decay of the eigenvalues of the data-free diagnostic matrix.

We note that the work of \cite{Baptista_Marzouk_Zahm_2022} extends this idea of data-averaged optimality (for the KL divergence only) by considering concurrent dimension reduction of both the state and the data variables, resulting in separate diagnostic matrices for each. Furthermore, they establish a tight connection to the problem of maximizing mutual information between the two random variables; we believe that similar results will hold for a ``mutual information'' defined with respect to $\alpha$-divergences using our results above. {This line of research has already been pursued by the authors of~\cite{cui4258736scalable} for the squared Hellinger distance. When considering dimension reduction in the state variables only, their bound~\cite[Prop. 1]{cui4258736scalable} coincides with the results of Theorem~\ref{corr:datafreebound}.}

\section{Extension to nonlinear feature detection}
\label{sec:nonlinear}

In this section, we propose a nonlinear extension of the proposed methodology: that is, we now seek a nonlinear feature map $\varphi_r:\R^d\rightarrow\R^r$ such that the target distribution $\pi$ can be well approximated (in the sense of $\alpha$-divergence) with a distribution $\pi_{\alpha,r}^\opt$ given as
$$
 \d\pi_{\alpha,r}^\opt(x) \propto  \ell_{\alpha,r}^\opt( \varphi_r(x) )\d\mu(x) ,
$$
where $\ell_{\alpha,r}^\opt$ is optimal profile function as given in Theorem \ref{thm:pythagorean}.
We consider feature maps of the form
$$
 \varphi_r(x) = U_r^\trans \Phi(x) ,
$$
where $U_r\in\R^{d\times r}$ is a matrix with orthogonal columns and $\Phi:\R^d\rightarrow\R^d$ is a $\mathcal{C}^1$-diffeomorphism.
{This diffeomorphic map performs a change-of-variables $z=\Phi(x)$ so that the feature map $\varphi_r$ corresponds to applying linear dimension reduction to the variable $z$. Denoting by}
$\Phi_\sharp \pi$ and $ \Phi_\sharp \pi_{\alpha,r}^\opt $  the pushforward distributions of $\pi$ and $\pi_{\alpha,r}^\opt$ through $\Phi$, respectively, the change-of-variables formula yields
\begin{align*}
 \d\Phi_\sharp \pi(z) &\propto \ell(\Phi^{-1}(z)) \d \Phi_\sharp\mu(z) \\
 \d\Phi_\sharp \pi_{\alpha,r}^\opt(z) &\propto \ell_{\alpha,r}^\opt( U_r^\trans z )\d \Phi_\sharp\mu(z) .
\end{align*}
{
Then the function $\ell_{\alpha,r}$ is explicitly given by
\begin{equation*}
 \ell_{\alpha,r}^\opt(\theta_r) =
 \begin{cases}
 \mathbb{E}_{Z\sim\Phi_\sharp\mu}[ \ell(\Phi^{-1}(Z))^\alpha \mid U_r^\trans Z=\theta_r ]^{\frac{1}{\alpha}} ,
      &  \alpha\neq0,  \\[5pt]
      \exp( \mathbb{E}_{Z\sim\Phi_\sharp\mu}[ \ln\ell(\Phi^{-1}(Z)) \mid U_r^\trans Z=\theta_r ] ),
      &   \alpha=0.
 \end{cases}
\end{equation*}
}
Evidently, the methodology presented in Section~\ref{sec:alphaCDR} directly applies by replacing the likelihood $x\mapsto \ell(x)$ with $z\mapsto \ell(\Phi^{-1}(z))$ and the reference $\mu$ with $\Phi_\sharp\mu$.
Since the $\alpha$-divergence is a $\phi$-divergence, it holds that
$
 D_\alpha(\pi||\pi_{\alpha,r}^\opt)
 = D_\alpha( \Phi_\sharp \pi || \Phi_\sharp \pi_{\alpha,r}^\opt ).
$
so that applying Theorem \ref{thm:alphaCDRext} yields
\begin{align*}
D_\alpha( \pi || \pi_{\alpha,r}^\opt)
&=D_\alpha( \Phi_\sharp \pi || \Phi_\sharp \pi_{\alpha,r}^\opt ) \\
&\overset{\eqref{eq:alphaCDRext}}{\leq} \mcJ_\alpha \lb C_{\min\{\frac{1}{\alpha};\,2\}}^\sub(\Phi_\sharp\mu)\, \E_{Z \sim \Phi_\sharp\pi}\left[ \|U_\perp^\trans \nabla \ln \ell\circ\Phi^{-1}(Z) \|_2^2 \right] \rb \\
&=\mcJ_\alpha \lb C_{\min\{\frac{1}{\alpha};\,2\}}^\sub(\Phi_\sharp\mu)\, \E_{X \sim  \pi}\left[ \|U_\perp^\trans \nabla \Phi(X)^{-\trans}\nabla \ln \ell( X ) \|_2^2 \right] \rb ,
\end{align*}
where $\nabla \Phi(x) = (\partial_j \Phi_i(x))_{1\leq i,j\leq d}$ denotes the Jacobian matrix of $\Phi$ evaluated at $x$.
As before, for a given $\Phi$, the optimal matrix $U_r$ from this upper bound contains the $r$ leading eigenvectors of the matrix
\begin{equation}\label{eq:defHPhi}
 H(\Phi) = \int \nabla \Phi(x)^{-\trans} \lb \nabla\log\ell(x) \nabla\log\ell(x)^\trans \rb \nabla \Phi(x) ^{-1} \d\pi(x) .
\end{equation}
Note that a similar object was considered in \cite{cui2022prior} for the purpose of preconditioning MCMC algorithms.
Denoting by $\lambda_k(\Phi)$ the $k$-th largest eigenvalue of $H(\Phi)$, this construction of $U_r$ provides the following error bound
\begin{equation}\label{eq:boundlossNL}
    D_\alpha( \pi || \pi_{\alpha,r}^\opt)
\leq
\mcJ_\alpha\lb C_{\min\{\frac{1}{\alpha};\,2\}}^\sub(\Phi_\sharp\mu) \sum_{k = r+1}^d \lambda_k (\Phi)\rb .
\end{equation}

This inequality suggests a construction for $\Phi$ based on two criteria. The first criterion is to control $C_{\min\{\frac{1}{\alpha};\,2\}}^\sub(\Phi_\sharp\mu)$, for instance by bounding it with some constant. By Proposition \ref{prop:CphiSubFinite}, one way to achieve such a bound is to impose $\Phi_\sharp\mu = \gamma$, where $\d\gamma(z)\propto\exp(-\|z\|^2/2)\d z$ is the standard Gaussian distribution on $\R^d$, thus ensuring $C_{\min\{\frac{1}{\alpha};\,2\}}^\sub(\Phi_\sharp\mu) \leq 1$,  Following the terminology used in machine learning, this suggests choosing $\Phi$ to be a \emph{normalizing flow} \cite{Tabak_Turner_2013, pmlr-v37-rezende15}. The second criterion is to optimize the decay in the spectrum of $H(\Phi)$, so that the right-hand side of \eqref{eq:boundlossNL} decays quickly with $r$. We propose to minimize the effective rank of $H(\Phi)$ defined by $\mathrm{erank}(H(\Phi))= \trace( H(\Phi))/\|H(\Phi)\|_\text{sp}$, where $\|\cdot\|_\text{sp}$ denotes the spectral norm, or equivalently,  $\mathrm{erank}(H(\Phi))=\sum_{k=1}^d \lambda_k(\Phi)/\lambda_1(\Phi) $. This is a continuous relaxation of the matrix rank, and we have $1\leq \mathrm{erank}\,H(\Phi) \leq \text{rank}\,H(\Phi)$, with equality if and only if $H(\Phi)$ is proportional to a projection matrix.
An alternative definition of the effective rank proposed in \cite{roy2007effective} enjoys similar properties and can equivalently be used here.

Accordingly, we propose constructing a $\Phi$ which accounts for both criteria as the solution to
\begin{equation}\label{eq:OptPhi}
 \min_{\substack{\Phi \in  \mathrm{Diff}(\R^d), \\ \Phi_\sharp\mu=\gamma }}
 \mathrm{erank}(H(\Phi)),
\end{equation}
where $\mathrm{Diff}(\R^d)$ denotes the set of $\mathcal{C}^1$-diffeomorphisms from $\R^d$ to $\R^d$.
Interestingly, this can be interpreted as a variant of the Monge problem, where $\mathrm{erank}(H(\Phi))$ replaces the transportation cost from $\mu$ to $\gamma$.
In effect, there are infinitely many diffeomorphisms $\Phi$ which satisfy $\Phi_\sharp\mu=\gamma$: the formulation \eqref{eq:OptPhi} seeks one that optimizes the decay of the spectrum of $H(\Phi)$. In principle, this construction could be replaced by a more general one that relaxes
the constraint $\Phi_\sharp\mu=\gamma$, for instance by allowing pushforward  distributions that are ``close'' to $\gamma$ in some sense, and thus exposing the trade-off between the spectral decay of $H(\Phi)$ and the subspace LSI constant of $\Phi_\sharp \mu$. The challenge with this more general construction is that it requires understanding the continuity of
$\nu\mapsto C_{\min\{\frac{1}{\alpha};\,2\}}^\sub(\nu)$ with respect to an appropriate metric/topology.

\begin{remark}
Consider the reference measure $\mu_\Sigma \sim \mathcal{N}(0, \Sigma)$ for some covariance $\Sigma \succ 0$. It is straightforward to show that $C_\beta^\sub(\mu_\Sigma) \leq \lambda_\textrm{max}(\Sigma)$, i.e., the largest eigenvalue of the covariance matrix~$\Sigma$. Applying Theorem~\ref{thm:alphaCDR} shows that the features~$U_r$ constructed from the dominant eigenvectors of the diagnostic matrix~$H$ provide the certificate of optimality~$\mcJ_\alpha(\lambda_\textrm{max}(\Sigma) \sum_{k > r} \lambda_k(H))$. On the other hand, if we precondition with the linear features $\Phi(x)=\Sigma^{-1/2}x$, then the optimal features~$U_r$ from the diagnostic matrix $H(\Phi)$ in~\eqref{eq:defHPhi} correspond instead
to the generalized eigenvector problem $Hv_i = \lambda_i\Sigma^{-1} v_i$.
Then, the new certificate of optimality from~\eqref{eq:boundlossNL} becomes $\mcJ_\alpha(\sum_{k > r} \lambda_k(H,\Sigma))$, where the $\lambda_k(H,\Sigma)$ now denotes the $k$-th generalized eigenvalue, and we used the fact that $C^\sub_\beta(\Phi_\sharp \mu_\Sigma)  \leq 1$ since $\Phi$ transforms $\mu_\Sigma$ to an isotropic Gaussian. In addition, the generalized eigenvalues satisfy $\lambda_k(H,\Sigma)\leq \lambda_{\max}(\Sigma)\lambda_k(H)$ for all $1\leq k\leq d$ so that $ \mcJ_\alpha(\sum_{k > r} \lambda_k(H,\Sigma)) \leq \mcJ_\alpha( \lambda_\textrm{max}(\Sigma) \sum_{k > r} \lambda_k(H))$. This shows that the preconditioning \emph{improves} the upper bound.
\end{remark}

\section{Connection to broader literature} \label{sec:litreview}

The main idea underlying this paper is the construction of low-dimensional ridge-based approximations of probability measures, using gradient information. To the best of our knowledge, this approximation class was first proposed in the context of \emph{likelihood informed subspaces} (LIS) by~\cite{Cui_Martin_Marzouk_Solonen_Spantini_2014}, for Bayesian inverse problems.
For regression and function approximation, we note that ridge approximations based on gradient information date back at least to the ``average derivative functionals'' of \cite{samarov1993exploring} and the notion of active subspaces \cite{constantine2015active}, which has since seen many refinements \cite{zahm2020gradient,lee2019modified} and wide application \cite{liu2023preintegration,constantine2017global}.
The LIS construction was made rigorous by Spantini et al.~\cite{SSC+15}, who proved its optimality in the linear Gaussian setting; \cite{ZCLSM22} later extended this optimality, in the sense of majorizations, to nonlinear forward models and certain non-Gaussian priors through their \emph{certified dimension reduction} approach. Among the many papers which have since appeared, we highlight in particular Cui and Tong~\cite{Cui_Tong_2021} as they analyze \eqref{eq:loss} for the squared Hellinger distance and the KL divergence, as well as the 2-Wasserstein metric. (See remarks in Section~\ref{sec:FirstBound} comparing our results for $d_{\text{Hell}}$ to those in \cite{Cui_Tong_2021}.) Their analysis also establishes a trade-off between defining the diagnostic matrix via integration over the reference measure $\mu$ (also as in \cite{Constantine_Kent_Bui-Thanh_2016}) rather than over the target measure $\pi$.

Our main technical tools for constructing our ridge-based approximations rely on functional Markov semigroup inequalities, which are traditionally used to characterize the long-time dynamics of Markov diffusion processes \cite{Bakry2014}. Interestingly, whereas the Poincar\'{e} inequality implies the exponential ergodicity of Langevin dynamics for the $\chi^2$-divergence, in this paper we use it to control the squared-Hellinger distance instead.

We note that these functional inequalities have also been considered for ridge-based approximation of probability measures by other authors. Notably, the work of \cite{Pillaud-Vivien_Bach_Lelievre_Rudi_Stoltz_2020} views the Poincar\'{e} constant of a measure as a proxy for its sampling complexity, and derives an algorithm to detect a low-dimensional subspace whose marginal distribution has maximal Poincar\'{e} constant. They highlight in particular applications in computational chemistry, where multi-modality and the discovery of suitable ``reaction coordinates'' are key challenges.
The idea of using low-dimensional subspaces for efficient sampling also appears in \cite{Chen_Ghattas_2020}, where a pre-computed LIS subspace is used in conjunction with Stein variational gradient descent (SVGD) in that subspace. \cite{brennan2020greedy} instead propose an iterative algorithm for approximating and sampling from high-dimensional distributions, using transport maps. At each step, the certified dimension reduction approach of \cite{ZCLSM22} is used to identify a low-dimensional subspace that best captures the departure of the target distribution from its current approximation, and an invertible transport map or normalizing flow is then constructed within this subspace. This construction recalls the iterative Gaussianization algorithm of \cite{laparra2011iterative}, but for the purpose of sampling a target whose unnormalized density is available and with subspaces identified via gradient information.
Liu et al.~\cite{Liu_Zhu_Ton_Wynne_Duncan_2022} provide an interesting extension of these ideas by using the kernelized Stein discrepancy to discover low-dimensional structure in the target measure, while concurrently sampling using SVGD. Intriguingly, they relate their optimal subspace to the spectrum of the matrix
\[
H_{\textsc{RKHS}} = \mathbb{E}_{X,\,X' \sim \mu}[\kappa(X,X') \nabla\ln\ell(X)\nabla\ln\ell(X')^\trans]
\]
for reproducing kernel $\kappa: \mathbb{R}^d \times \mathbb{R}^d \to \mathbb{R}$ and prior~$\mu$. For the (inadmissible) choice of kernel $\kappa(x,y) = \delta(x-y)$, one recovers the diagnostic matrix $H$ as in \eqref{eq:defH}, but with an expectation taken over $\mu$ instead of over $\pi$. We believe further extensions of kernelization and Bayesian dimension reduction to be an interesting direction for future research.

The framework proposed in \cite{ZCLSM22} has also been used to make rare event simulation more tractable in~\cite{uribe2021cross,tong2022large}. In this context, the function $\ell$ is the indicator of a \emph{failure domain}, bounded by the level set of a given performance function. Replacing this indicator function with a smooth approximation (e.g., a sigmoid function) yields a regular likelihood function whose gradients are well-defined. These gradients are then used to reduce the dimension of the problem to facilitate sampling of the rare event.

Lastly, we mention that the dimensional bottleneck for MCMC can be circumvented by exploiting alternative structural assumptions on the target measure.
In the setting of Bayesian inverse problems with Gaussian priors, the pre-conditioned Crank-Nicolson (pCN) sampler of \cite{Cotter_Roberts_Stuart_White_2013} yields discretization-invariant and hence dimension-independent MCMC performance when the unknown parameters represent a discretization of an underlying function. Yet the performance of vanilla pCN can scale poorly with other aspects of the target, e.g., concentration relative to the prior. Hence many efforts have focused on modifying pCN to account for the target geometry \cite{rudolf2018generalization,beskos2017geometric,NoemiTOMS}; the latter include the dimension-independent likelihood-informed MCMC algorithms of \cite{Cui_Law_Marzouk_2016}, which explicitly use low-dimensional ridge approximations of the form studied in this paper.

\section{Conclusion}

In this paper, we develop gradient-based dimension reduction algorithms for probability measures. Notably, we provide certifiable approximation guarantees for our algorithms by establishing a connection with functional inequalities from Markov semigroup theory. This extends earlier work by showing that dimension reduction quantified by the Amari $\alpha$-divergences, for $\alpha \in (0,1]$, can be certified using the corresponding $\nicefrac{1}{\alpha}$-Sobolev inequalities.

It would be of interest to derive certifiable dimension reduction algorithms for other probability metrics; for example, \cite{Cui_Tong_2021} also analyze the 2-Wasserstein metric. One approach could be to use the Talagrand $T_2$ functional inequality, which bounds the 2-Wasserstein metric by the KL divergence  \cite[Ch.~9]{Bakry2014} and thus would enable us to certify dimension reduction in the 2-Wasserstein metric using the KL certificate of optimality. But it is of interest whether more direct arguments can be used to obtain tighter certificates.
(We have not yet attempted to characterize $\pi^\opt_{W_2,r}$; our Theorem~\ref{thm:pythagorean} cannot be applied directly, since it exploits in some sense the property that $\alpha$-divergences are essentially Bregman divergences.)

In deriving our improved majorization function for $\alpha \in [\nicefrac{1}{2}, 1]$ in Theorem~\ref{thm:alphaCDRimp}, we benefited from recent improvements to functional inequalities, which are still an active research field.  Bolley and Gentil also present a refinement for the $\beta$-Sobolev inequalities for $\beta \in (2,4)$ \cite[Theorem 17]{Bolley_Gentil_2010}, which suggests that we may obtain tighter dimension reduction certificates for $\alpha$-divergences with $\alpha \in (\nicefrac{1}{4}, \nicefrac{1}{2})$. However, the resulting majorization function is nonlinear with respect to the matrix $U_r$, and to the best of our understanding does not admit a closed form solution. We leave exploring this, and other possible improvements, to future work.

{As one of our anonymous referees astutely noted, the $\beta$-Sobolev inequalities we consider herein
are in fact \emph{equivalent} to the Poincar\'e inequality when $\beta>1$, and to the LSI when $\beta=1$.
It is then natural to ask whether other functional inequalities which interpolate the Poincar\'e inequality and the LSI differently could be used to further strengthen our results.
For example, the Lata\l{}a–Oleszkiewicz (LO) inequalities~\cite{Latala_Oleszkiewicz_2000} interpolate between these two inequalities in the sense of capturing a range of different \emph{tail behaviors} of measures.
It is an interesting open question whether these inequalities (or others) can be applied to our problem, and whether doing so might capture a deeper interplay between the choice of $\alpha$-divergence and assumptions on the tails of the reference measure. However, in comparison with the $\beta$-Sobolev inequalities, we note that this exercise may not yield majorizations which are easily minimized to discover the optimal features.

The anonymous referee also raised the question of the extent to which tail properties of the reference measure should be related to the choice of $\alpha$-divergence for dimension reduction at all. As an example, consider~\cite{Vempala_Wibisono_2019, Chewi_Erdogdu_Li_Shen_Zhang_2021}, in which assumptions of the LO inequality (with parameter $\gamma \in [0,1]$) for the target measure are used to demonstrate exponential ergodicity of unadjusted Langevin Monte Carlo for \emph{every} Renyi divergence. The interplay of $\gamma$, which captures the analog of the tail behaviour referenced above, and the choice of Renyi divergence manifests in the strength of the discretized convergence rates. It is therefore an interesting direction of future work to explore whether analogous results can be established for our context.
}

Lastly, we note that readers familiar with both the Poincar\'e inequality and the logarithmic Sobolev inequality may wonder if the $\phi$-Sobolev inequalities similarly satisfy a tensorization property. Chafa\"i \cite[\S 3.1]{Chafai_2004} proves that indeed they do, and therefore they possess the same powerful dimension-free scaling as the other functional inequalities. On the other hand, in the applications we consider here the dimension~$d$ is fixed, and one might wonder whether this philosophical incompatibility introduces deficits into our upper bounds. Indeed, in forthcoming work we show that by using \emph{dimensional} Sobolev inequalities \cite{Bakry_Ledoux_2006} we can provide tighter certificates of optimality for the Kullback--Leiber divergence, as well as for the squared Hellinger distance.

\subsection*{Funding}
ML and YMM acknowledge support from the US Department of Energy, Office of Advanced Scientific Computing Research, under grants DE-SC0023187 and DE-SC0023188, and from the ExxonMobil Technology and Engineering Company.
OZ acknowledges support from the ANR JCJC project MODENA (ANR-21-CE46-0006-01).

{
\bibliographystyle{siam}
\bibliography{references}
}

\begin{appendix}

\section{Proofs}

\subsection{Proof of Theorem~\ref{thm:pythagorean} }
\label{sec:pythagoreanproof}

\begin{proof}
\noindent\textbf{Case $\alpha \notin\{0;1\}$.} Let $\d\pi_{\alpha,r}^\opt = Z_{\alpha,r}^{-1}\,\ell^\opt_{\alpha,r}\circ\varphi_r\d\mu$ as in \eqref{eq:piopt} and \eqref{eq:opt_llhd}. For any approximate measure of the form $\d\widetilde\pi_r = \widetilde{Z}^{-1} \widetilde\ell_r\circ \varphi_r  \d\mu$, we have
\begin{align}
\alpha(\alpha-1)D_\alpha(\pi \,||\, \widetilde\pi_r) &\overset{\eqref{eq:DalphaDef}}{=}
\frac{\widetilde{Z}^{\alpha-1}}{Z_\pi^\alpha} \int \ell^\alpha ( \widetilde\ell_r\circ\varphi_r)^{1-\alpha} \d\mu - 1 \nonumber\\
&\stackrel{\eqref{eq:condExp2}}{=}
\frac{\widetilde{Z}^{\alpha-1}}{Z_\pi^\alpha} \int \left(\int \ell^\alpha\d\mu_{\perp|r}\right) (\widetilde\ell_r)^{ 1-\alpha} \d\mu_r - 1 \nonumber\\
&\stackrel{\eqref{eq:opt_llhd}}{=}\frac{\widetilde{Z}^{\alpha-1}}{Z_\pi^\alpha} \int (\ell_{\alpha,r}^\opt)^\alpha (\widetilde\ell_r)^{ 1-\alpha} \d\mu_r - 1 \nonumber\\
&=\lb\frac{Z_{\alpha,r}}{Z_\pi}\rb^{\alpha} \lb \int \lb \frac{\ell_{\alpha,r}^\opt}{Z_{\alpha,r}}\rb^\alpha \, \lb\frac{\widetilde\ell_r}{\widetilde{Z}}\rb^{1-\alpha} \d\mu_r - 1 \rb + \lb\frac{Z_{\alpha,r}}{Z_\pi}\rb^{\alpha}  -1 \nonumber\\
&\overset{\eqref{eq:DalphaDef}}{=}
\alpha(\alpha-1)\lb\frac{Z_{\alpha,r}}{Z_\pi}\rb^{\alpha} D(\pi_{\alpha,r}^\opt \,||\, \widetilde\pi_r) +  \lb\frac{Z_{\alpha,r}}{Z_\pi}\rb^{\alpha}  -1.
\label{eq:tmp4928589317942}
\end{align}
If we let $\widetilde\ell_r=\ell_{\alpha,r}^\opt$ in the above equation, we deduce that
$$D_\alpha(\pi||\pi_{\alpha,r}^\opt)= \frac{1}{\alpha(\alpha-1)}\left( \left(\frac{Z_{\alpha,r}}{Z_\pi} \right)^{\alpha}  -1\right).$$
Combining the above identity with \eqref{eq:tmp4928589317942} re-scaled by $\alpha(\alpha-1)$, we obtain \eqref{eq:pythagorean}.

\noindent \textbf{Case $\alpha = 1$.} The original proof for this choice of $\alpha$, which corresponds to the KL divergence, was originally shown in~\cite{ZCLSM22}. We include the details here for completeness. From~\eqref{eq:opt_llhd} we have
$
\d \pi_{1,r}^\opt(\theta_r) = \frac{1}{Z_{1,r}} \left( \int \ell(x)\d\mu_{\perp|r}(x|\theta_r) \right)\d \mu_r(\theta_r).
$
Integrating against $\mu_r$ yields $1=Z_\pi/Z_{1,r}$ so that $Z_{1,r}=Z_\pi$. Then we can write
\begin{align*}
D_1(\pi\,||\,\widetilde\pi_r) &=
\int \ln \lb \frac{Z_\pi^{-1} \ell}{\widetilde{Z}^{-1}\widetilde\ell_r}\rb \frac{\ell}{Z_\pi}\d\mu \\
&=\int \ln \lb \frac{Z_\pi^{-1} \ell}{Z_\pi^{-1}\ell_{1,r}^\opt\circ\varphi_r}\rb \frac{\ell}{Z_\pi}\d\mu  +
\int \ln \lb \frac{Z_\pi^{-1} \ell_{1,r}^\opt\circ\varphi_r}{\widetilde{Z}^{-1}\widetilde\ell_r\circ\varphi_r}\rb \frac{\ell}{Z_\pi}\d\mu  \\
&\overset{\eqref{eq:condExp2}}{=}\int \ln \lb \frac{Z_\pi^{-1} \ell}{Z_\pi^{-1}\ell_{1,r}^\opt\circ\varphi_r}\rb \frac{\ell}{Z_\pi}\d\mu  +
\int \ln \lb \frac{Z_\pi^{-1} \ell_{1,r}^\opt }{\widetilde{Z}^{-1}\widetilde\ell_r }\rb  \frac{ \int \ell \d\mu_{\perp|r} }{Z_\pi}\d\mu_r  \\
&\overset{\eqref{eq:opt_llhd}}{=}\int \ln \lb \frac{Z_\pi^{-1} \ell}{Z_\pi^{-1}\ell_{1,r}^\opt\circ\varphi_r}\rb \frac{\ell}{Z_\pi}\d\mu  +
\int \ln \lb \frac{Z_\pi^{-1} \ell_{1,r}^\opt }{\widetilde{Z}^{-1}\widetilde\ell_r }\rb  \frac{\ell_{1,r}^\opt }{Z_\pi}\d\mu_r  \\
&= D_1(\pi \,||\, \pi_{1,r}^\opt) + D_1(\pi_{1,r}^\opt \,||\, \widetilde\pi_r),
\end{align*}
which is \eqref{eq:pythagorean}.

\noindent\textbf{Case $\alpha = 0$.} The choice $\alpha=0$ corresponds to the reverse KL divergence. For $\d\widetilde\pi_r = \widetilde{Z}^{-1} \widetilde\ell_r\circ \varphi_r \d\mu$ and $\d\pi_{0,r}^\opt = Z_0^{-1} \ell_{0,r}^\opt\circ \varphi_r \d\mu$ we have
\begin{align}
D_0(\pi\,||\,\widetilde\pi_r)
&= \int \ln\left( \frac{\widetilde Z^{-1} \widetilde\ell_r\circ\varphi_r}{Z_\pi^{-1} \ell}\right)  \frac{ \widetilde\ell_r\circ\varphi_r}{\widetilde Z}\d\mu \nonumber\\
&\overset{\eqref{eq:condExp2}}{=} \ln \frac{Z_\pi}{\widetilde Z} + \frac{1}{\widetilde Z} \int   \widetilde\ell_r \ln \widetilde\ell_r   \d\mu_r - \frac{1}{\widetilde Z} \int   \widetilde\ell_r \left(\int \ln \ell \d\mu_{\perp|r}\right)  \d\mu_r \nonumber\\
&\overset{\eqref{eq:opt_llhd}}{=} \ln \frac{Z_\pi}{\widetilde Z} + \frac{1}{\widetilde Z} \int   \widetilde\ell_r \ln \widetilde\ell_r   \d\mu_r - \frac{1}{\widetilde Z} \int   \widetilde\ell_r \ln \ell^\opt_{0,r}  \d\mu_r \nonumber\\
&=\ln \frac{Z_\pi}{Z_{0,r}} +  \int  \ln \left(\frac{\widetilde Z^{-1}\widetilde\ell_r}{Z_{0,r}^{-1}\ell^\opt_{0,r}} \right) \frac{ \widetilde\ell_r }{\widetilde Z} \d\mu_r \nonumber\\
&=\ln \frac{Z_\pi}{Z_{0,r}} +  D_0(\pi_{0,r}^\opt\,||\,\widetilde\pi_r). \label{eq:tmp325078}
\end{align}
If we let $\widetilde\ell_r=\ell_{0,r}^\opt$ in the above equation, we obtain
$$D_0(\pi||\pi_{0,r}^\opt)=\ln(Z_\pi/Z_{0,r}).$$
Combining the above identity with \eqref{eq:tmp325078}, we obtain \eqref{eq:pythagorean}. This completes the proof.
\end{proof}

\subsection{Proof of Proposition~\ref{prop:KLisQuasiOptimal} }
\label{proof:KLisQuasiOptimal}

\begin{proof}
 For any $0<\alpha\leq1$, applying Jensen's inequality yields $\ell_{\alpha,r}^\opt(\theta_r) \leq \ell_{1,r}^\opt(\theta_r)$.
 Therefore, we have
 \begin{align*}
  \alpha(\alpha-1) D_\alpha \lb \pi || \pi_{1,r}^\opt  \rb
  &\overset{\eqref{eq:DalphaDef}}{=} \int \left(\frac{\d\pi}{\d\pi_{1,r}^\opt} \right)^\alpha \d\pi_{1,r}^\opt- 1 \\
  &\overset{\eqref{eq:KLoptimalDecomp}}{=} \frac{1}{Z_\pi}\int \frac{\int\ell^\alpha\d\mu_{\perp|r}}{(\ell_{1,r}^\opt)^{\alpha-1}}  \d\mu_r- 1 \\
  &\overset{\eqref{eq:opt_llhd}}{=} \frac{1}{Z_\pi} \int \frac{(\ell_{\alpha,r}^\opt)^\alpha}{(\ell_{1,r}^\opt)^{\alpha-1}}  \d\mu_r- 1 \\
  &\geq \frac{1}{Z_\pi} \int \frac{(\ell_{\alpha,r}^\opt)^\alpha}{(\ell_{\alpha,r}^\opt)^{\alpha-1}}  \d\mu_r- 1
  = \frac{Z_\alpha}{Z_\pi} - 1.
 \end{align*}
 Then we obtain
 \begin{align*}
  D_\alpha \lb \pi || \pi_{1,r}^\opt  \rb
  &\leq \frac{1}{\alpha(\alpha-1)}\left(\frac{Z_\alpha}{Z_\pi} - 1\right)
  \overset{\eqref{eq:optdivergence}}{=} \frac{1}{\alpha(\alpha-1)}\left( \left( 1+\alpha(\alpha-1) D_\alpha \lb \pi || \pi_{\alpha,r}^\opt  \rb \right)_+^{1/\alpha}  - 1 \right).
 \end{align*}
 Because $(1+t)_+^{1/\alpha} \geq 1+t/\alpha$ whenever $1/\alpha\geq1$, we deduce \eqref{eq:KLisQuasiOptimal} and conclude the proof.
\end{proof}

\subsection{Proof of Proposition~\ref{prop:BoundDalphaNoDR}}
\begin{proof}
 For $\alpha=1$, the result is a straightforward application of the logarithmic Sobolev inequality, see, e.g., \cite{Gross_1975}.
 Assume $\nicefrac{1}{2} \leq\alpha <1$.
 Applying \eqref{eq:betaSI} with $\alpha=\nicefrac{1}{\beta}$ and $f = (Z_\pi^{-1}\ell)^\alpha$ yields
 \begin{equation*}
  \frac{1}{\frac{1}{\alpha}(\frac{1}{\alpha}-1)} \lb 1-\lb \int (Z_\pi^{-1}\ell)^\alpha \d\mu\rb^{\frac{1}{\alpha}}  \rb \leq \frac{C_{1/\alpha}(\mu)}{2} \mathbb{E}_{X \sim \pi}[ \alpha^2 \| \nabla \ln \ell(X)\|_2^2 ] ,
 \end{equation*}
 which after re-arranging recovers
 \begin{equation}\label{eq:tpm1356089}
  \int (Z_\pi^{-1}\ell)^\alpha \d\mu
  \geq
   \lb 1- \frac{(1-\alpha)}{2} C_{1/\alpha}(\mu)\, \mathbb{E}_{X \sim \pi}[ \| \nabla \ln \ell(X)\|_2^2 ] \rb_{+}^\alpha .
 \end{equation}
 Thus, the divergence $D_\alpha(\pi||\mu)$ satisfies
 $$
  D_\alpha(\pi\,||\,\mu) \overset{\eqref{eq:DalphaDef}}{=}
  \frac{1}{\alpha(\alpha-1)} \lb    \int (Z_\pi^{-1}\ell)^\alpha \d\mu  -1 \rb
  \overset{\eqref{eq:tpm1356089}}{\leq}
  \mathcal{J}_\alpha \lb C_{1/\alpha}(\mu)\,\mathbb{E}_{X \sim \pi}[ \| \nabla \ln \ell(X)\|_2^2] \rb .
 $$
 This shows \eqref{eq:BoundDalphaNoDR} and concludes the proof.
\end{proof}

\subsection{Proof of Proposition~\ref{prop:CphiSubFinite}}

\begin{proof}
For this class of measure, the Bakry--Emery criterion~\cite{Bakry2014} combined with the Holley--Stroock perturbation argument~\cite{Holley_Stroock_1987} lets one bound the logarithmic Sobolev constant as $C_1(\mu)\leq \exp(\sup B - \inf B)/R$.
As shown in~\cite[\S3.3]{Chafai_2004}, it also implies $C_\phi(\mu)\leq \exp(\sup B - \inf B)/R$ for any convex functions $\phi$ as in the definition of the $\phi$-Sobolev inequality.
To show
\begin{equation}
 C^\mathrm{sub}_\phi( \mu ) \leq \frac{\exp(\sup B - \inf B)}{R}   ,
\end{equation}
it is sufficient to observe that conditioning on $U_r^TX=\theta_r$ preserves the geometric structure of $\mu$ so that $\mu_{\perp|r}$ becomes
$$
 \d \mu_{\perp|r}(x|\theta_r) \propto \exp(-V_{\perp|r}(x|\theta_r)-B_{\perp|r}(x|\theta_r)) \d x_\perp ,
$$
where $\d x_\perp$ denotes the Lebesgue measures on the affine subspace $\{U_r \theta_r + U_\perp \theta_\perp \mid \theta_\perp\in\R^{d-r}\}$.
Here, $V_{\perp|r}(x|\theta_r)=V(U_r \theta_r  + U_\perp U_\perp^\trans x)$ and $B_{\perp|r}(x|\theta_r) = B(U_r \theta_r  + U_\perp U_\perp^\trans x)$ are respectively such that
$\Hess V_{\perp|r}(x|\theta_r) \succeq U_\perp^\trans \Hess(V) U_\perp \succeq RI_{d-r}$
and
$\sup B_{\perp|r}(\cdot|\theta_r) - \inf B_{\perp|r}(\cdot|\theta_r) \leq \sup B - \inf B$.
This yields $C_\phi(\mu_{\perp|r})\leq \exp(\sup B - \inf B))/R$ and concludes the proof.
\end{proof}

\subsection{Proof of Theorem~\ref{thm:alphaCDR}}
\begin{proof}
 The specific case of $\alpha=1$ follows from~\cite{ZCLSM22}.
 For $\alpha\in[\nicefrac{1}{2},1)$, applying \eqref{eq:betaSI} for the conditional measure $\mu_{\perp|r}(\cdot|\theta_r)$ with $\alpha=\nicefrac{1}{\beta}$ and $f = (Z_\pi^{-1}\ell)^\alpha$ being defined on the affine subspace $U_r\theta_r + \text{range}(U_\perp)$ yields
 \begin{align*}
  &\frac{1}{\frac{1}{\alpha}(\frac{1}{\alpha}-1)} \lb \int (Z_\pi^{-1}\ell) \d\mu_{\perp|r}(\cdot \mid \theta_r) -\lb \int (Z_\pi^{-1}\ell)^\alpha \d\mu_{\perp|r}(\cdot\mid \theta_r)\rb^{\frac{1}{\alpha}}  \rb \\
  &\leq \frac{C_{1/\alpha}(\mu_{\perp|r}(\cdot | \theta_r))}{2}\, {\int \,\alpha^2 \| U_\perp^\trans\nabla \ln \ell \|_2^2 \frac{\ell}{Z_\pi}\d\mu_{\perp \mid r}(\cdot \mid \theta_r)}.
 \end{align*}
 Integrating the above inequality with respect to $\theta_r$ against the measure $\mu_r$ yields
 \begin{equation*}
  \frac{1}{(1-\alpha)} \lb 1 - \int \lb \int (Z_\pi^{-1}\ell)^\alpha \d\mu_{\perp|r}\rb^{\frac{1}{\alpha}} \d\mu_r \rb
  \leq \frac{C^\mathrm{sub}_{1/\alpha}( \mu )}{2}\, \mathbb{E}_{X \sim \pi}\left[ \|  U_\perp^\trans \nabla \ln \ell(X) \|_2^2 \right]  ,
 \end{equation*}
 where we use the fact that $C_{1/\alpha}(\mu_{\perp|r}(\cdot|\theta_r)) \leq C^\mathrm{sub}_{1/\alpha}( \mu )$.
 Recalling the normalizing constant is $Z_{\alpha,r} = \int ( \int \ell^\alpha\d\mu_{\perp|r} )^{\frac{1}{\alpha}} \d\mu_r$, we can write
  $$
  \frac{Z_{\alpha,r}}{Z_\pi}
  \geq
   \lb 1- \frac{(1-\alpha)}{2} C^\mathrm{sub}_{1/\alpha}( \mu ) \mathbb{E}_{X \sim \pi}\left[ \|  U_\perp^\trans \nabla \ln \ell(X) \|_2^2\right] \rb_{+} ,
 $$
 which yields
 $$
  D_\alpha \lb \pi \,||\, \pi_{\alpha,r}^\opt\, \rb
  \overset{\eqref{eq:optdivergence}}{=} \dfrac{1}{\alpha(\alpha-1)}\lb \lb \dfrac{Z_{\alpha,r}}{Z_\pi}\rb^\alpha - 1 \rb
  \leq \mathcal{J}_\alpha\lb C^\mathrm{sub}_{1/\alpha}( \mu ) \mathbb{E}_{X \sim \pi}\left[ \|  U_\perp^\trans \nabla \ln \ell(X) \|_2^2\right] \rb .
 $$
 This concludes the proof.
\end{proof}

\subsection{Proof of Theorem~\ref{thm:alphaCDRimp}}
\begin{proof}
For $\beta\in (1,2)$ apply~\eqref{eq:betaSIimproved} for the conditional measure $\mu_{\perp | r}(\cdot | \theta_r)$ defined on $\varphi_r^{-1}(\theta_r) = \{U_r \theta_r + \textrm{range}(U_\perp)\}$. This yields
\begin{equation}
\label{eq:subphiSIimp}
\frac{1}{(\beta-1)^2} \lb \mu_{\perp | r}(f^\beta) - \mu_{\perp | r}(f)^{2\beta -2} \mu_{\perp | r}(f^\beta)^{\frac{2}{\beta}-1} \rb \leq C_\beta^\sub(\mu) \int f^\beta \|U_\perp^\trans \nabla_x \ln f \|_2^2 \d\mu_{\perp | r}
\end{equation}
for sufficiently smooth positive functions~$f:\mathbb{R}^d \to \mathbb{R}_{\geq 0}$.
As before, the principal idea is to compute the expectation of~\eqref{eq:subphiSIimp} under the marginalized law of~$\Theta_r \sim \mu_r$. However, this approach is obstructed by the nonlinear term in the left-hand side. We bound this product using H\"older's inequality to obtain
\[
\mu_r \lb \mu_{\perp | r}(f)^{2\beta -2} \mu_{\perp | r}(f^\beta)^{\frac{2}{\beta}-1} \rb  \leq
\lb \mu_r \circ \mu_{\perp | r}(f)^{p(2\beta -2)} \rb^\frac{1}{p}
\lb \mu_r \circ \mu_{\perp | r}(f^\beta)^{q(\frac{2}{\beta}-1)} \rb^\frac{1}{q}
\]
for generic conjugate exponents $p,\,q \in [1,\infty]$ with $\nicefrac{1}{p} + \nicefrac{1}{q} = 1$. In particular, we select $\frac{1}{q} = \frac{2}{\beta}-1$, which by conjugacy implies $p = \frac{\beta}{2\beta - 2}$, and since $\beta \in (1,2)$ both choices of $p$ and $q$ are admissible. This results in the inequality
\[
\frac{1}{(\beta-1)^2} \lb \mu(f^\beta) - \mu_r \left[(\mu_{\perp | r}(f))^\beta \right]^\frac{1}{p} \mu(f^\beta)^\frac{1}{q} \rb \leq C^\sub_\beta(\mu) \,\mathbb{E}_\mu \left[ f^\beta \|U_\perp^\trans \nabla_x \ln f \|_2^2 \right].
\]
Substituting the test function $f = (Z_\pi^{-1}\ell)^\frac{1}{\beta}$, we have
\begin{equation}
\frac{1}{(\beta-1)^2} \lb 1 - (Z_\pi^{-1} \mu_r [\mu_{\perp | r}(\ell^\frac{1}{\beta})]^\beta )^{2 - \frac{2}{\beta}} \rb \leq \frac{C^\sub_\beta}{\beta^2}\, \mathbb{E}_\pi \left[ \| U_\perp^\trans \nabla_x \ln \ell \|^2_2 \right]
\end{equation}
after applying the change of measure $\d\pi=Z_\pi^{-1}\ell\d\mu$.
Re-arranging the inequality and recalling $Z_{\alpha,r} = \mu_r[(\mu_{\perp | r}(\ell^\alpha))^\frac{1}{\alpha}]$ obtains
\begin{equation}
\label{eq:norm_lower_bound_refined}
\frac{Z_{\alpha,r}}{Z_\pi} \geq \lb 1 - C^\sub_{1/\alpha}(\mu) \lb 1 - \alpha \rb^2 \mathbb{E}_\pi [\| U_\perp^\trans \nabla_x \ln \ell \|^2_2 ] \rb^\frac{1}{2-2\alpha} .
\end{equation}
By~\eqref{eq:optdivergence}, we obtain the upper bound \eqref{eq:alphaCDRimp} on $D_\alpha(\pi||\pi^\opt_{\alpha,r})$ for $\alpha \in [\nicefrac{1}{2}, 1]$.  The bound for $\alpha \in (0, \nicefrac{1}{2}]$ is derived identically, except using the bound $C_{1/\alpha}^\sub(\mu) < \frac{1}{2\alpha}C_2^\sub(\mu)$ for the $\nicefrac{1}{\alpha}$-Sobolev constant.
\end{proof}

\section{Comparison of $\mcJ_\alpha$ and $\mcJ_\alpha^\flat$}
\label{sec:comparison}

For $\nicefrac{1}{2} \leq \alpha < 1$ we have the majorized loss function
\begin{equation*}
\mcJ_\alpha(u) = \frac{1}{\alpha(\alpha-1)} \left[ \lb 1 - \frac{(1-\alpha)}{2}u \rb_+^\alpha - 1 \right] .\end{equation*}
Its derivatives are given by
\begin{align*}
\mcJ_\alpha'(u) &= \frac{1}{2} \lb 1 - \frac{(1-\alpha)}{2} u \rb_+^{\alpha-1}, \\
\mcJ_\alpha^{''}(u) &= \frac{1}{4}(1-\alpha)^2 \lb 1 - \frac{(1-\alpha)}{2} u \rb_+^{\alpha-2} ,
\end{align*}
for all $u\neq 2/(1-\alpha)$.
Evidently, $\mcJ_\alpha$ is monotone non-decreasing on $\R_{\geq 0}$ and it is convex on the intervals $[0,2/(1-\alpha) )$ and $( 2/(1-\alpha), +\infty)$, but not on the whole line $\R_{\geq0}$.

We recall that the improved majorized loss function  is given by
\begin{equation*}
\mcJ_\alpha^\flat(u) = \frac{1}{\alpha(\alpha-1)} \left[ \lb 1 - (1-\alpha)^2 u \rb_+^{\frac{\alpha}{2(1-\alpha)}} - 1 \right]
\end{equation*}
for $\nicefrac{1}{2} \leq \alpha \leq 1$.
Its derivatives are given by
\begin{align*}
 (\mcJ_\alpha^\flat)'(u) &= \frac{1}{2} \lb 1 - (1-\alpha)^2 u\rb_+^{\frac{\alpha}{2(1-\alpha)} -1}, \\
 (\mcJ_\alpha^\flat)''(u) &=  \frac{1}{4} (3\alpha-2)(\alpha-1)\lb 1 - (1-\alpha)^2 u\rb_+^{\frac{\alpha}{2(1-\alpha)}-2},
\end{align*}
for all $u\neq (1-\alpha)^{-2}$. Therefore, $\mcJ_\alpha^\flat$ is also monotone non-decreasing  on $\R_{\geq0}$ and, if $\alpha \geq \nicefrac{2}{3}$, then $\mcJ_\alpha^\flat$ is \emph{concave} on the whole $\R_{\geq0}$ (this can be proven by showing that $\mcJ^\beta(u)$ is always bellow its tangents).

We provide a comparison lemma between these majorization functions; see also Figure~\ref{fig:comparison2} for a visual confirmation of the lemma.
\begin{lemma} For all $\nicefrac{1}{2} \leq \alpha < 1$ we have $\mcJ^\flat_\alpha(u) \leq \mcJ_\alpha(u)$ for all $u \geq 0$. Conversely, for all $0 < \alpha < \nicefrac{1}{2}$ it holds that $\mcJ_\alpha(u) \leq \mcJ^\flat_\alpha(u)$ for all $u \geq 0$.
\end{lemma}
\begin{proof}
First consider $\alpha  \in [\nicefrac{1}{2}, 1)$ and let $u_\textrm{crit}^\flat = \frac{1}{(1-\alpha)^2}$ and $u_\textrm{crit} = \frac{2}{1-\alpha}$ denote the values after which $\mcJ^\flat(u)$ and $\mcJ(u)$ saturate to the limit $|\alpha(\alpha-1)|^{-1}$, respectively. Since $u_\textrm{crit} < u^\flat_\textrm{crit}$ for $\alpha \geq \nicefrac{1}{2}$, by monotonicity it suffices to consider $0 < u < u_\textrm{crit}$. Over this interval, we note that $\mcJ_\alpha^\flat \leq \mcJ_\alpha$ is equivalent to
\[
(1 - (1-\alpha)^2 u)^\frac{1}{2(1-\alpha)} \geq  1 - \frac{1-\alpha}{2}u .
\]
But this follows immediately from Bernoulli's inequality, since $\frac{1}{2(1-\alpha)} \geq 1$.

Now consider $\alpha \in (0, \nicefrac{1}{2})$ and analogously define $v^\flat_\text{crit} = \frac{2\alpha}{(1-\alpha)^2}$ and $v_\text{crit} = \frac{4\alpha}{1-\alpha}$. Since $v^\flat_\text{crit} < v_\text{crit}$, by monotonicity it suffices to show that $\mcJ_\alpha \leq \mcJ_\alpha^\flat$ over $0 < u < v^\flat_\text{crit}$. This is equivalent to demonstrating
\[
\lb 1 - \frac{(1-\alpha)^2}{2\alpha}u \rb^\frac{1}{2(1-\alpha)} \leq 1 - \frac{1-\alpha}{4\alpha}u,
\]
which again follows from Bernoulli's inequality as we have $\frac{1}{2(1-\alpha)} < 1$.
\end{proof}

\begin{figure}[h!]
\label{fig:Jalpha_vs_J_02}
\begin{centering}
\includegraphics[width=\textwidth]{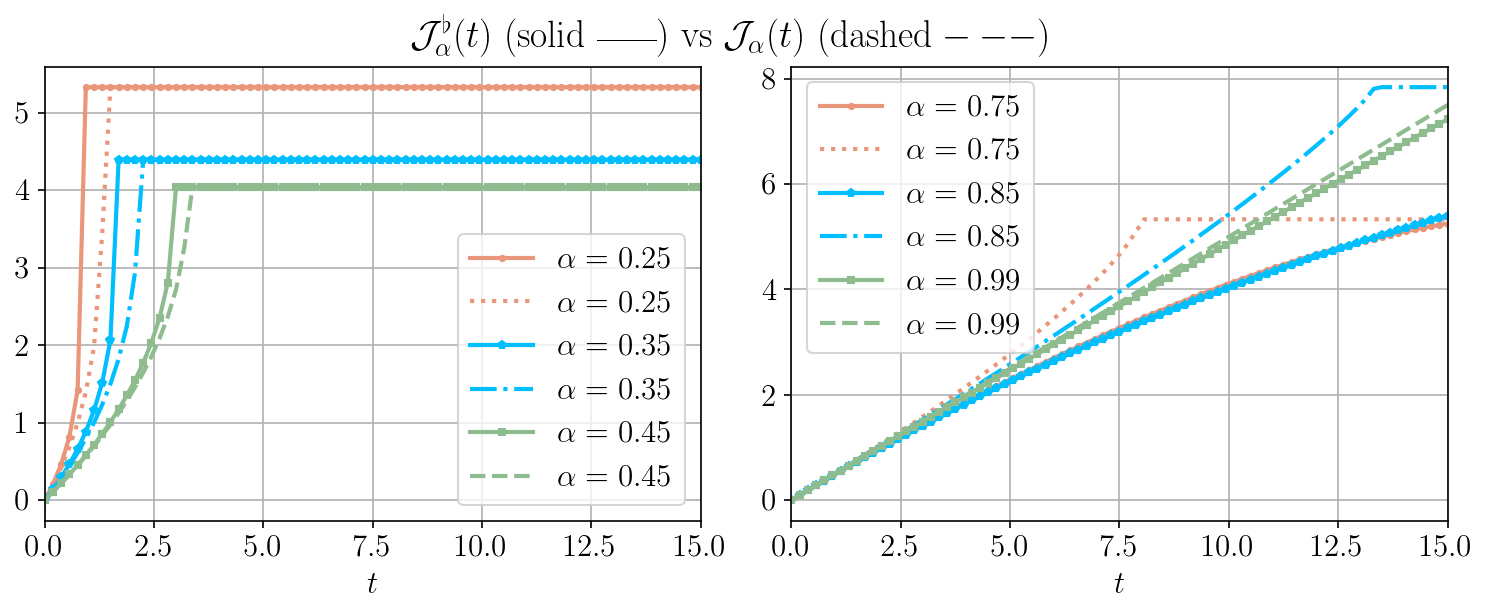}
\end{centering}
\caption{ \label{fig:comparison2} Comparison between the improved majorized loss function  $\mcJ^\flat_\alpha(t)$ and the majorized loss function $\mcJ_\alpha(t)$ for (left) several choices of $\alpha < \nicefrac{1}{2}$, and (right) several choices of $\alpha > \nicefrac{1}{2}$.
}
\end{figure}

\section{Analytical results for linear-Gaussian inverse problems}
\label{sec:lingaussian}

We consider an example problem for which~$\pi^\opt_{\alpha,r}$ (i.e., $\ell^\opt_{\alpha,r}$), the diagnostic matrix~$H$, and the divergence ~$D_\alpha(\pi \,|| \pi^\opt_{\alpha,r})$ can be computed analytically. Let $X$ be distributed according to prior $\mu = \mcN(0, I_d)$, and consider a conditional likelihood distribution $Y|X=x$ which is multivariate Gaussian with mean $Ax$ and covariance $I_m$, for some matrix $A \in \R^{m \times d}$ and $m < d$. We then have the log likelihood function
\[
x \mapsto \ln \ell(x) = -\frac{1}{2} \| y - Ax \|_2^2.
\]
For ease of calculation assume the realization $y=0$ for the data. Then, the posterior measure has Lebesgue density
\[
\d\pi(x) \propto \exp\lb -\frac{1}{2} x^\trans (I_d + A^\trans A) x \rb \,\d x,
\]
i.e., it is multivariate Gaussian with zero mean and covariance $\Sigma = (I_d + A^\trans A)^{-1}$.

For this example the diagnostic matrix is given by
\[
H = \E_\pi[ \nabla \ln \ell \,\nabla \ln \ell^\trans] = A^\trans A (I_d + A^\trans A)^{-1} A^\trans A.
\]
Let $\gamma_k$ denote the $k$-th eigenvalue of $A^\trans A$. Then, by functional calculus we know that $\lambda_k = \frac{\gamma_k^2}{1+\gamma_k}$ is the $k$-th eigenvalue of $H$; as will be convenient later on, we have equivalently that
\begin{equation}
\label{eq:gammalbd}
\gamma_k = \frac{1}{2}(\lambda_k + \sqrt{\lambda_k}\sqrt{\lambda_k+4}).
\end{equation}
Furthermore, both $A$ and $H$ share the same eigenvectors; we let $U_r \in \R^{d \times r}$ denote the matrix whose columns correspond to the first $r$ eigenvectors.

For this choice of $U_r$ we have the decomposition $x = U_r\theta_r + U_\perp\theta_\perp$, where $U_\perp$ is any orthogonal completion to $U_r$. Furthermore, the measure $\mu_{\perp | r}$ is simply a multivariate standard Gaussian in $d-r$ dimensions. Consider $\alpha \neq 0$. Then, according to Theorem~2.1 in the main text we have
\begin{align*}
\ell^\opt_{\alpha,r}&(\theta_r) = \E_{X \sim \mu}[ \ell(X)^\alpha \mid  U_r^\trans X = \theta_r]^{\frac{1}{\alpha}}
\simeq \exp\lb -\frac{1}{2} \| AU_r\theta_r\|_2^2 \rb
\end{align*}
(Note that $\simeq$ indicates hidden multiplicative constants that may vary line-to-line. In order to apply the result in the theorem it would be critical to keep track of these quantities. However, for the two divergences we consider below we use a more convenient expression for $D_\alpha(\pi \,||\, \pi^\opt_{\alpha,r})$.)
This shows that for all $\alpha \neq 0$ we have
\[
\ell^\opt_{\alpha,r}(x) \simeq \exp\lb -\frac{1}{2} x^\trans U_rU_r^\trans A^\trans A U_rU_r^\trans x \rb = \exp \lb -\frac{1}{2} x^\trans U_r \begin{pmatrix} \Gamma_r &  \\  & 0_{d-r} \end{pmatrix} U_r^\trans x\rb ,
\]
where $\Gamma_r = \textrm{diag}(\gamma_1, \ldots, \gamma_r)$.  Since this is in the form of an exponential, this implies that $\pi^\opt_{\alpha,r}$ is multivariate Gaussian with zero mean and covariance
\[
\Sigma^\opt_{\alpha,r} = U \begin{pmatrix} (I_r + \Gamma_r)^{-1} & \\ & I_{d-r} \end{pmatrix} U^\trans,
\]
where $U = \begin{bmatrix} U_r & U_\perp \end{bmatrix}$.\\

\noindent\textbf{Kullback-Leiber Divergence.}
The computations for the Kullback-Leiber divergence can also be found in \cite[\S 2.3]{ZCLSM22}; for completeness we include them here as well. Since both $\pi$ and $\pi^\opt_{\alpha,r}$ are multivariate Gaussian, we have that
\[
D_\textsc{KL}(\pi\,||\,\pi^\opt_{\alpha,r}) = \frac{1}{2} (\textrm{trace}( (\Sigma^\opt_{\alpha,r})^{-1} \Sigma) - \ln \det( (\Sigma^\opt_{\alpha,r})^{-1} \Sigma) ) - d).
\]
Since we have the explicit eigenvalue decomposition of both covariance matrices, we can equivalently express this as
\[
D_\textsc{KL}(\pi\,||\,\pi^\opt_{\alpha,r}) = \frac{1}{2} \sum_{k > r}^d \lb \ln (1+\gamma_k) - \frac{\gamma_k}{1+\gamma_k}\rb,
\]
which can be related to the diagnostic matrix using \eqref{eq:gammalbd}.\\

\noindent\textbf{Squared Hellinger distance.}
Similarly, the squared Hellinger distance between two multivariate Gaussians is known in closed form. For our example, this yields
\[
d^2_\textsc{Hell}(\pi, \pi^\opt_{\alpha,r}) = 1 - \det(\Sigma^\opt_{\alpha,r})^{\frac{1}{4}} \det(\Sigma)^{\frac{1}{4} }
\det\lb \frac{1}{2}(\Sigma^\opt_{\alpha,r} + \Sigma) \rb^{-\frac{1}{2}}.
\]
These computations simplify using the eigenvalue decomposition of both covariance matrices, and we obtain
\[
d^2_\textsc{Hell}(\pi, \pi^\opt_{\alpha,r}) = 1- \lb \prod_{k > r} \frac{1+\gamma_k}{(1+\frac{1}{2}\gamma_k)^2} \rb^\frac{1}{4},
\]
which can be related to the diagnostic matrix using \eqref{eq:gammalbd}.

\end{appendix}


\end{document}